\documentclass[11pt,a4paper]{llncs}
\usepackage[english]{babel}
\usepackage{german}
\usepackage{amssymb}
\usepackage{amsfonts}
\usepackage{amsmath}
\usepackage[latin1]{inputenc}
\usepackage{fullpage}
\usepackage{graphicx}
\usepackage{subfigure}
\usepackage{algorithm}
\usepackage[noend]{algorithmic}
\usepackage{amsmath,amssymb,cite}
\usepackage{lineno,footmisc,marvosym}
\usepackage{wasysym,vmargin,stackrel}
\usepackage{color,hyperref}

\newtheorem{observation}{Observation}

\newtheorem{reductionRule}{Reduction Rule}
\setmarginsrb{3.5cm}{2.25cm}{3.5cm}{3.05cm}{0.3cm}{0.3cm}{-0.3cm}{1.0cm}
   \renewenvironment{thebibliography}[1]{
 \begin{oldthebibliography}{#1}
 \setlength{\parskip}{0.0ex} \setlength{\itemsep}{1ex}}  {\end{oldthebibliography}}

\begin{document}

\title{Algorithms and Almost Tight Results \\
for $3$-Colorability of Small Diameter Graphs}
\author{George B.~Mertzios\inst{1}\thanks{%
Partially supported by EPSRC Grant EP/G043434/1.} \and Paul G.~Spirakis%
\inst{2}\thanks{%
Partially supported by the ERC EU Project RIMACO and by the EU IP FET
Project MULTIPLEX.}}
\institute{School of Engineering and Computing Sciences, Durham University, UK. \\
Email: \email{george.mertzios@durham.ac.uk} \and
Computer Technology Institute and University of Patras, Greece. \\
Email: \email{spirakis@cti.gr}\vspace{-0.2cm}}
\maketitle

\begin{abstract}
The $3$-coloring problem is well known to be NP-complete. It is also well
known that it remains NP-complete when the input is restricted to graphs
with diameter~$4$. Moreover, assuming the Exponential Time Hypothesis (ETH), 
$3$-coloring cannot be solved in time $2^{o(n)}$ on graphs with~$n$ vertices
and diameter at most~$4$. In spite of the extensive studies of the $3$%
-coloring problem with respect to several basic parameters, the complexity
status of this problem on graphs with small diameter, i.e.~with diameter at
most~$2$, or at most~$3$, has been a longstanding and challenging open
question. In this paper we investigate graphs with small diameter. For
graphs with diameter at most~$2$, we provide the first subexponential
algorithm for $3$-coloring, with complexity $2^{O(\sqrt{n\log n})}$, which
is asymptotically the same as the currently best known time complexity for
the graph isomorphism problem. Furthermore we extend the notion of an
articulation vertex to that of an \emph{articulation neighborhood}, and we
provide a polynomial algorithm for $3$-coloring on graphs with diameter~$2$
that have at least one articulation neighborhood. 
For graphs with diameter at most~$3$, we establish the complexity of $3$%
-coloring, even for the case of triangle-free graphs. Namely we prove that
for every~${\varepsilon \in \lbrack 0,1)}$, $3$-coloring is NP-complete on
triangle-free graphs of diameter~$3$ and radius $2$ with $n$ vertices and
minimum degree $\delta =\Theta (n^{\varepsilon })$. Moreover, assuming~ETH,
we use three different amplification techniques of our hardness results, in
order to obtain for every~${\varepsilon \in \lbrack 0,1)}$ subexponential
asymptotic lower bounds for the complexity of $3$-coloring on triangle-free
graphs with diameter~$3$ and minimum degree~${\delta =\Theta (n^{\varepsilon
})}$. Finally, we provide a $3$-coloring algorithm with running time~${%
2^{O(\min \{\delta \Delta ,\ \frac{n}{\delta }\log \delta \})}}$ for
arbitrary graphs with diameter~$3$, where $n$ is the number of vertices and $%
\delta $ (resp.~$\Delta $) is the minimum (resp.~maximum) degree of the
input graph. To the best of our knowledge, this algorithm is the first
subexponential algorithm for graphs with ${\delta =\omega (1)}$ and for
graphs with ${\delta =O(1)}$ and $\Delta =o(n)$. Due to the above lower
bounds of the complexity of $3$-coloring, the running time of this algorithm
is asymptotically almost tight when the minimum degree of the input graph is 
$\delta =\Theta (n^{\varepsilon })$, where $\varepsilon \in \lbrack \frac{1}{%
2},1)$. \newline

\end{abstract}

\section{Introduction\label{intro-sec}}

A \emph{proper }$k$\emph{-coloring} (or $k$\emph{-coloring}) of a graph $G$
is an assignment of $k$ different colors to the vertices of $G$, such that
no two adjacent vertices receive the same color. That is, a $k$-coloring is
a partition of the vertices of $G$ into $k$ independent sets. The
corresponding $k$\emph{-coloring problem} is the problem of deciding whether
a given graph $G$ admits a $k$-coloring of its vertices, and to compute one
if it exists. Furthermore, the minimum number $k$ of colors for which there
exists a $k$-coloring is denoted by $\chi (G)$ and is termed the \emph{%
chromatic number} of $G$. The \emph{minimum coloring problem} is to compute
the chromatic number of a given graph $G$, and to compute a $\chi (G)$%
-coloring of $G$.

One of the most well known complexity results is that the $k$-coloring
problem is NP-complete for every $k\geq 3$, while it can be solved in
polynomial time for $k=2$~\cite{GareyJohnson79}. Therefore, since graph
coloring has numerous applications besides its theoretical interest, there
has been considerable interest in studying how several graph parameters
affect the tractability of the $k$-coloring problem, where $k\geq 3$. In
view of this, the complexity status of the coloring problem has been
established for many graph classes. It has been proved that $3$-coloring
remains NP-complete even when the input graph is restricted to be a line
graph~\cite{Holyer81}, a triangle-free graph with maximum degree $4$~\cite%
{Maffray96}, or a planar graph with maximum degree $4$~\cite{GareyJohnson79}.

On the positive side, one of the most famous result in this context has been
that the minimum coloring problem can be solved in polynomial time for
perfect graphs using the ellipsoid method~\cite{Grotschel84}. Furthermore,
polynomial algorithms for $3$-coloring have been also presented for classes
of non-perfect graphs, such as AT-free graphs~\cite{Stacho-3-col-AT-free-11}
and $P_{6}$-free graphs~\cite{Randerath04} (i.e.~graphs that do not contain
any path on $6$ vertices as an induced subgraph). Furthermore, although the
minimum coloring problem is NP-complete on $P_{5}$-free graphs, the $k$%
-coloring problem is polynomial on these graphs for every fixed $k$~\cite%
{Hoang10}. Courcelle's celebrated theorem states that every problem
definable in Monadic Second-Order logic (MSO) can be solved in linear time
on graphs with bounded treewidth~\cite{CourcelleMSO90}, and thus also the
coloring problem can be solved in linear time on such graphs.

For the cases where $3$-coloring is NP-complete, considerable attention has
been given to devise exact algorithms that are faster than the brute-force
algorithm (see e.g.~the recent book~\cite{FominKratsch-book-10}). In this
context, asymptotic lower bounds of the time complexity have been provided
for the main NP-complete problems, based on the \emph{Exponential Time
Hypothesis (ETH)}~\cite{Impagliazzo-ETH-01,Impagliazzo-Strongly-ETH-01}. ETH
states that there exists no deterministic algorithm that solves the $3$SAT
problem in time $2^{o(n)}$, given a boolean formula with $n$ variables. In
particular, assuming ETH, $3$-coloring cannot be solved in time~$2^{o(n)}$
on graphs with $n$ vertices, even when the input is restricted to graphs
with diameter~$4$ and radius $2$ (see~\cite%
{Survey-ETH-11,PapadimitriouComplexity94}). Therefore, since it is assumed
that no subexponential $2^{o(n)}$ time algorithms exist for $3$-coloring,
most attention has been given to decrease the multiplicative factor of $n$
in the exponent of the running time of exact exponential algorithms, see
e.g.~\cite{FominKratsch-book-10,Beigel05,NarayanaswamyIPL11}.

One of the most central notions in a graph is the distance between two
vertices, which is the basis of the definition of other important
parameters, such as the diameter, the eccentricity, and the radius of a
graph. For these graph parameters, it is known that $3$-coloring is
NP-complete on graphs with diameter at most $4$ (see e.g.~the standard proof
of~\cite{PapadimitriouComplexity94}). Furthermore, it is straightforward to
check that $k$-coloring is NP-complete for graphs with diameter at most $2$,
for every $k\geq 4$: we can reduce $3$-coloring on arbitrary graphs to $4$%
-coloring on graphs with diameter~$2$, just by introducing to an arbitrary
graph a new vertex that is adjacent to all others.

In contrast, in spite of the extensive studies of the $3$-coloring problem
with respect to several basic parameters, the complexity status of this
problem on graphs with small diameter, i.e.~with diameter at most $2$ or at
most $3$, has been a longstanding and challenging open question, see e.g.~%
\cite{Kaminskyi-open-11,Barnaby-survey-11,Broersma09}. The complexity status
of $3$-coloring is open also for triangle-free graphs of diameter~$2$ and of
diameter~$3$. It is worth mentioning here that a graph is triangle-free and
of diameter~$2$ if and only if it is a maximal triangle free graph.
Moreover, it is known that $3$-coloring is NP-complete for triangle-free
graphs~\cite{Maffray96}, however it is not known whether this reduction can
be extended to maximal triangle free graphs. Another interesting result is
that almost all graphs have diameter~$2$~\cite{Bollobas81}; however, this
result cannot be used in order to establish the complexity of $3$-coloring
for graphs with diameter~$2$.

\vspace{0.2cm} \noindent\textbf{Our contribution.} In this paper we provide
subexponential algorithms and hardness results for the $3$-coloring problem
on graphs with low diameter, i.e.~with diameter~$2$ and $3$. As a
preprocessing step, we first present two reduction rules that we apply to an
arbitrary graph $G$, such that the resulting graph $G^{\prime }$ is $3$%
-colorable if and only $G$ is $3$-colorable. We call the resulting graph 
\emph{irreducible} with respect to these two reduction rules. We use these
reduction rules to reduce the size of the given graph and to simplify the
algorithms that we present.

For graphs with diameter at most $2$, we first provide a subexponential
algorithm for $3$-coloring with running time ${2^{O(\min \{\delta ,\frac{n}{%
\delta }\log \delta \})}}$, where $n$ is the number of vertices and $\delta $
is the minimum degree of the input graph. This algorithm is simple and has
worst-case running time $2^{O(\sqrt{n\log n})}$, which is asymptotically the
same as the currently best known time complexity of the graph isomorphism
problem~\cite{Babai83}. To the best of our knowledge, this algorithm is the
first subexponential algorithm for graphs with diameter~$2$. We demonstrate
that this is indeed the worst-case of our algorithm by providing, for every $%
n\geq 1$, a $3$-colorable graph $G_{n}=(V_{n},E_{n})$ with $\Theta (n)$
vertices, such that~$G_{n}$ has diameter~$2$ and both its minimum degree and
the size of a minimum dominating set is~$\Theta (\sqrt{n})$. In addition,
this graph is triangle-free and irreducible with respect to the above two
reduction rules. Furthermore we define the notion of an \emph{articulation
neighborhood} in a graph, which extends the notion of an articulation
vertex. For any vertex $v$ of a graph $G$, the closed neighborhood $N[v] =
N(v) \cup \{v\}$ is an articulation neighborhood if $G-N[v]$ is
disconnected. We present a polynomial algorithm for every irreducible graph $%
G$ with diameter~$2$, which has at least one articulation neighborhood.

For graphs with diameter at most $3$, we establish the complexity of
deciding $3$-coloring, even for the case of triangle-free graphs. Namely we
prove that $3$-coloring is NP-complete on irreducible and triangle-free
graphs with diameter~$3$ and radius $2$, by providing a reduction from 3SAT.
In addition, we provide a $3$-coloring algorithm with running time~${%
2^{O(\min \{\delta \Delta ,\ \frac{n}{\delta }\log \delta \})}}$ for
arbitrary graphs with diameter~$3$, where $n$ is the number of vertices and $%
\delta $ (resp.~$\Delta $) is the minimum (resp.~maximum) degree of the
input graph. To the best of our knowledge, this algorithm is the first
subexponential algorithm for graphs with ${\delta =\omega (1)}$ and for
graphs with ${\delta =O(1)}$ and ${\Delta =o(n)}$. Table~\ref%
{coloring-complexities-table} summarizes the current state of the art of the
complexity of $k$-coloring, as well as our algorithmic and
NP-completeness~results.

\vspace{-0.3cm}

\begin{table}[htb]
\begin{center}
\begin{tabular}{|c|c|c|c|}
\hline
\ \ $k \ \backslash \ diam(G)$ \ \  & \ \ $2$ \ \  & \ \ $3$ \ \  & \ \ $%
\geq 4$ \ \  \\ \hline
\ \ \textcolor{white}{$3$} \ \  & \ \ \textcolor{white}{$\frac{%
\frac{1}{1}}{1}$} \ \  & \ \ ($\ast$) NP-complete for \ \  & \ \ %
\textcolor{white}{bla} \ \  \\ 
\ \ $3$ \ \  & \ \ ($\ast$) ${2^{O(\min\{\delta,\frac{n}{\delta}%
\log\delta\})}}$-time \ \  & \ \ minimum degree $\delta =
\Theta(n^{\varepsilon})$, \ \  & \ \ \textcolor{white}{bla} \ \  \\ 
\ \ \textcolor{white}{$3$} \ \  & \ \ algorithm \ \  & \ \ %
\textcolor{white}{$\frac{}{\frac{}{}}$} for every $\varepsilon \in [0,1)$,
even if \textcolor{white}{$\frac{\frac{}{}}{}$} \ \  & \ \ NP-complete~\cite%
{PapadimitriouComplexity94} \ \  \\ 
\ \ \textcolor{white}{$3$} \ \  & \ \ ($\ast$) polynomial algorithm \ \  & \
\ $rad(G) = 2$ and $G$ is triangle-free \ \  & \ \ and no $2^{o(n)}$-time
algorithm \ \  \\ 
\ \ \textcolor{white}{$3$} \ \  & \ if there is at least one \  & \ \ ($\ast$%
) ${2^{O(\min\{\delta\Delta,\ \frac{n}{\delta}\log\delta\})}}$-time %
\textcolor{white}{$\frac{\frac{.}{\frac{.}{.}}}{.}$} \ \  & \ \ %
\textcolor{white}{bla} \ \  \\ 
\ \ \textcolor{white}{$3$} \ \  & \ articulation neighborhood \  & \ \
algorithm \ \  & \ \ \textcolor{white}{bla} \ \  \\ \hline
\ \ $\geq 4$ \ \  & \ \ NP-complete \ \  & \ \ NP-complete \ \  & \ \
NP-complete \ \  \\ \hline
\end{tabular}
\vspace{0.1cm}
\end{center}
\caption{Current state of the art and our algorithmic and NP-completeness
results for $k$-coloring on graphs with diameter~$diam(G)$. Our results are
indicated by an asterisk ($\ast$).\protect\vspace{-0.9cm}}
\label{coloring-complexities-table}
\end{table}

Furthermore, we provide three different amplification techniques that extend
our hardness results for graphs with diameter~$3$. In particular, we first
show that $3$-coloring is NP-complete on irreducible and triangle-free
graphs $G$ of diameter~$3$ and radius~$2$ with $n$ vertices and minimum
degree $\delta (G)=\Theta (n^{\varepsilon })$, for every $\varepsilon \in
\lbrack \frac{1}{2},1)$ and that, for such graphs, there exists no algorithm
for $3$-coloring with running time~${2^{o(\frac{n}{\delta }%
)}=2^{o(n^{1-\varepsilon })}}$, assuming ETH. This lower bound is
asymptotically almost tight, due to our above algorithm with running time~${%
2^{O(\frac{n}{\delta }\log \delta )}}$, which is subexponential when $\delta
(G)=\Theta (n^{\varepsilon })$ for some $\varepsilon \in \lbrack \frac{1}{2}%
,1)$. With our second amplification technique, we show that $3$-coloring
remains NP-complete also on irreducible and triangle-free graphs $G$ of
diameter~$3$ and radius~$2$ with $n$ vertices and minimum degree~${\delta
(G)=\Theta (n^{\varepsilon })}$, for every $\varepsilon \in \lbrack 0,\frac{1%
}{2})$. Moreover, we prove that for~such graphs, when ${\varepsilon \in
\lbrack 0,\frac{1}{3})}$, there exists no algorithm for $3$-coloring
with~running time~${2^{o(\sqrt{\frac{n}{\delta }})}=2^{o(n^{(\frac{%
1-\varepsilon }{2})})}}$, assuming ETH. Finally, with our third
amplification technique, we prove that for such graphs, when ${\varepsilon
\in \lbrack \frac{1}{3},\frac{1}{2})}$, there exists no algorithm for $3$%
-coloring with running~time~${2^{o(\delta )}=2^{o(n^{\varepsilon })}}$,
assuming ETH. Table~\ref{lower-bounds-table} summarizes our lower time
complexity bounds for $3$-coloring on irreducible and triangle-free graphs
with diameter~$3$ and radius~$2$, parameterized by their minimum degree $%
\delta$.

\begin{table}[htb]
\begin{center}
\begin{tabular}{|c|c|c|c|}
\hline
\textcolor{white}{$\frac{}{\frac{}{}}$} $\delta(G) = \Theta(n^{\varepsilon})$%
: \textcolor{white}{$\frac{\frac{\frac{}{}}{}}{}$} & $0 \leq \varepsilon < 
\frac{1}{3}$ & $\frac{1}{3} \leq\varepsilon < \frac{1}{2}$ & $\frac{1}{2}
\leq \varepsilon < 1$ \\ \hline
\hspace{-0.2cm} \textcolor{white}{$\frac{\frac{\frac{\frac{}{.}}{.}}{.}}{.}$}
Lower time \textcolor{white}{$\frac{\frac{\frac{.}{.}}{.}}{.}$
\hspace{-0.2cm}} & \hspace{0.1cm} no $2^{o(n^{(\frac{1-\varepsilon}{2})})}$%
-time algorithm \hspace{0.1cm} & \hspace{0.1cm} no $2^{o(n^{\varepsilon})}$%
-time algorithm \hspace{0.1cm} & \hspace{0.1cm} no $2^{o(n^{1-\varepsilon})}$%
-time algorithm \hspace{0.1cm} \\ 
\textcolor{white}{$\frac{}{\frac{.}{.}}$} complexity bound: %
\textcolor{white}{$\frac{\frac{.}{.}}{}$} & (cf.~Theorem~\ref%
{ETH-small-min-degree-thm}) & (cf.~Theorem~\ref{ETH-medium-min-degree-thm})
& (cf.~Theorem~\ref{ETH-large-min-degree-thm}) \\ \hline
\end{tabular}
\vspace{0.1cm}
\end{center}
\caption{Summary of the results of Theorems~\protect\ref%
{ETH-large-min-degree-thm},~\protect\ref{ETH-medium-min-degree-thm}, and~%
\protect\ref{ETH-small-min-degree-thm}: Lower time complexity bounds for
deciding $3$-coloring on irreducible and triangle-free graphs $G$ with $n$
vertices, diameter~$3$, radius $2$, and minimum degree ${\protect\delta(G) =
\Theta(n^{\protect\varepsilon})}$, where ${\protect\varepsilon \in [0,1)}$,
assuming~ETH. The lower bound for ${\protect\varepsilon \in [\frac{1}{2}, 1)}
$ is asymptotically almost tight, as there exists an algorithm for arbitrary
graphs with diameter $3$ with running time~${2^{O(\frac{n}{\protect\delta }%
\log \protect\delta)} = 2^{O(n^{1-\protect\varepsilon} \log n)}}$ by Theorem~%
\protect\ref{min-delta-Delta-n-over-delta-alg-diam-3-thm}.\protect\vspace{%
-0.7cm}}
\label{lower-bounds-table}
\end{table}

\vspace{0.0cm} \noindent \textbf{Organization of the paper.} We provide in
Section~\ref{preliminaries-sec} the necessary notation and terminology, as
well as our two reduction rules and the notion of an irreducible graph. In
Sections~\ref{diameter-two-sec} and~\ref{diameter-three-sec} we present our
results for graphs with diameter~$2$ and $3$, respectively. 

\section{Preliminaries and notation\label{preliminaries-sec}}

In this section we provide some notation and terminology, as well as two
reduction (or \textquotedblleft cleaning\textquotedblright ) rules that can
be applied to an arbitrary graph $G$. Throughout the article, we assume that
any given graph $G$ of low diameter is \emph{irreducible} with respect to
these two reduction rules, i.e.~that these reduction rules have been
iteratively applied to $G$ until they cannot be applied any more. Note that
the iterative application of these reduction rules on a graph with $n$
vertices can be done in time polynomial in $n$.

\medskip

\textbf{Notation.} We consider in this article simple undirected graphs with
no loops or multiple edges. In an undirected graph $G$, the edge between
vertices $u$ and $v$ is denoted by~$uv$, and in this case $u$ and $v$ are
said to be \emph{adjacent} in $G$. Otherwise $u$ and $v$ are called \emph{%
non-adjacent} or \emph{independent}. Given a graph ${G=(V,E)}$ and a vertex $%
{u\in V}$, denote by ${N(u)=\{v\in V : uv\in E\}}$ the set of neighbors (or
the \emph{open neighborhood}) of $u$ and by~${N[u]=N(u)\cup \{u\}}$ the 
\emph{closed neighborhood} of $u$. Whenever the graph $G$ is not clear from
the context, we will write $N_{G}(u)$ and $N_{G}[u]$, respectively. Denote
by $\deg (u)=|N(u)|$ the \emph{degree} of $u$ in $G$ and by $\delta (G)=\min
\{\deg (u) : u\in V\}$ the \emph{minimum degree} of~$G$. Let $u$ and $v$ be
two non-adjacent vertices of~$G$. Then, $u$ and $v$ are called (false) \emph{%
twins} if they have the same set of neighbors, i.e.~if $N(u)=N(v)$.
Furthermore, we call the vertices $u$ and $v$ \emph{siblings} if ${%
N(u)\subseteq N(v)}$ or ${N(v)\subseteq N(u)}$; note that two twins are
always siblings.

Given a graph ${G=(V,E)}$ and two vertices $u,v\in V$, we denote by $d(u,v)$
the \emph{distance} of~$u$ and~$v$, i.e.~the length of a shortest path
between $u$ and $v$ in $G$. Furthermore, we denote by $diam(G)=\max \{d(u,v)
: u,v\in V\}$ the \emph{diameter} of $G$ and by $rad(G)=\min_{u\in V}\{\max
\{d(u,v) : v\in V\}\}$ the \emph{radius} of $G$. Given a subset ${S\subseteq
V}$, $G[S]$ denotes the \emph{induced} subgraph of $G$ on the vertices in $S$%
. We denote for simplicity by $G-S$ the induced subgraph $G[V\setminus S]$
of $G$. A subset $S\subseteq V$ is an \emph{independent set} in $G$ if the
graph $G[S]$ has no edges. If $G[S]$ has all ${\binom{|S|}{2}}$ possible
edges among its vertices, then~$G[S]$ is a \emph{clique}. A clique with $t$
vertices is denoted by $K_{t}$. A graph $G$ that contains no~$K_{t}$ as an
induced subgraph is called $K_{t}$\emph{-free}. Furthermore, a subset $%
D\subseteq V$ is a \emph{dominating set} of~$G$ if every vertex of $%
V\setminus D$ has at least one neighbor in $D$. 
For simplicity, we refer in the remainder of the article to a proper $k$%
-coloring of a graph $G$ just as a $k$\emph{-coloring} of $G$. Throughout
the article we perform several times the \emph{merging} operation of two (or
more) independent vertices, which is defined as follows: we \emph{merge} the
independent vertices~$u_{1},u_{2},\ldots ,u_{t}$ when we replace them by a
new vertex $u_{0}$ with $N(u_{0})=\cup _{i=1}^{k}N(u_{i})$. In addition to
the well known big-$O$ notation for asymptotic complexity, some times we use
the $O^{\ast }$ notation that suppresses polynomially bounded factors. For
instance, for functions $f$ and $g$, we write $f(n)=O^{\ast }(g(n))$ if $%
f(n)=O(g(n)\ poly(n))$, where $poly(n)$ is a polynomial.

\medskip

Observe that, whenever a graph~$G$ contains a clique~$K_{4}$ with four
vertices as an induced subgraph, then~$G$ is not~$3$-colorable. Furthermore,
we can check easily in polynomial time (e.g.~with brute-force) whether a
given graph~$G$ contains a~$K_{4}$. Therefore we assume in the following
that all given graphs are~$K_{4}$-free. Furthermore, since a graph is $3$%
-colorable if and only if all its connected components are $3$-colorable, we
assume in the following that all given graphs are connected. In order to
present our two reduction rules of an arbitrary~$K_{4}$-free graph~$G$,
recall first that the \emph{diamond} graph is a graph with~$4$ vertices and~$%
5$ edges, i.e.~it consists of a~$K_{4}$ without one edge. The diamond graph
is illustrated in Figure~\ref{diamond-fig}. Suppose that four vertices~$%
u_{1},u_{2},u_{3},u_{4} $ of a given graph~$G=(V,E)$ induce a diamond graph,
and assume without loss of generality that~${u_{1}u_{2}\notin E}$. Then, it
is easy to see that in any~$3$-coloring of~$G$ (if such exists), $u_{1}$ and~%
$u_{2}$ obtain necessarily the same color. Therefore we can merge~$u_{1}$
and~$u_{2}$ into one vertex, as the next reduction rule states, and the
resulting graph is $3$-colorable if and only if $G$ is $3$-colorable.

\begin{reductionRule}[diamond elimination]
\label{rule1}Let ${G=(V,E)}$ be a $K_{4}$-free graph. If the quadruple ${%
\{u_{1},u_{2},u_{3},u_{4}\}}$ of vertices in $G$ induces a diamond graph,
where ${u_{1}u_{2}\notin E}$, then merge vertices~$u_{1}$ and~$u_{2}$.
\end{reductionRule}

Note that, after performing a diamond elimination in a $K_{4}$-free graph $G$%
, we may introduce a new $K_{4}$ in the resulting graph. An example of such
a graph $G$ is illustrated in Figure~\ref{diamond-elimination-new-K4-fig}.
In this example, the graph on the left hand side has no $K_{4}$ but it has
two diamonds, namely on the quadruples ${\{u_{1},u_{2},u_{3},u_{4}\}}$ and ${%
\{u_{4},u_{5},u_{6},u_{7}\}}$ of vertices. However, after eliminating the
first diamond by merging $u_{1}$ and $u_{4}$, we create a new $K_{4}$ on the
quadruple ${\{u_{1},u_{5},u_{6},u_{7}\}}$ of vertices, cf.~the graph of the
right hand side of Figure~\ref{diamond-elimination-new-K4-fig}.

\vspace{-0.2cm}

\begin{figure}[h!tb]
\centering 
\subfigure[]{ \label{diamond-fig}
\includegraphics[scale=0.6]{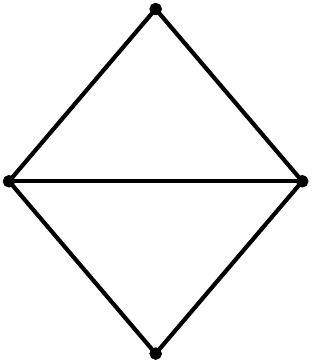}} \hspace{1.3cm} 
\subfigure[]{ \label{diamond-elimination-new-K4-fig}
\includegraphics[scale=0.6]{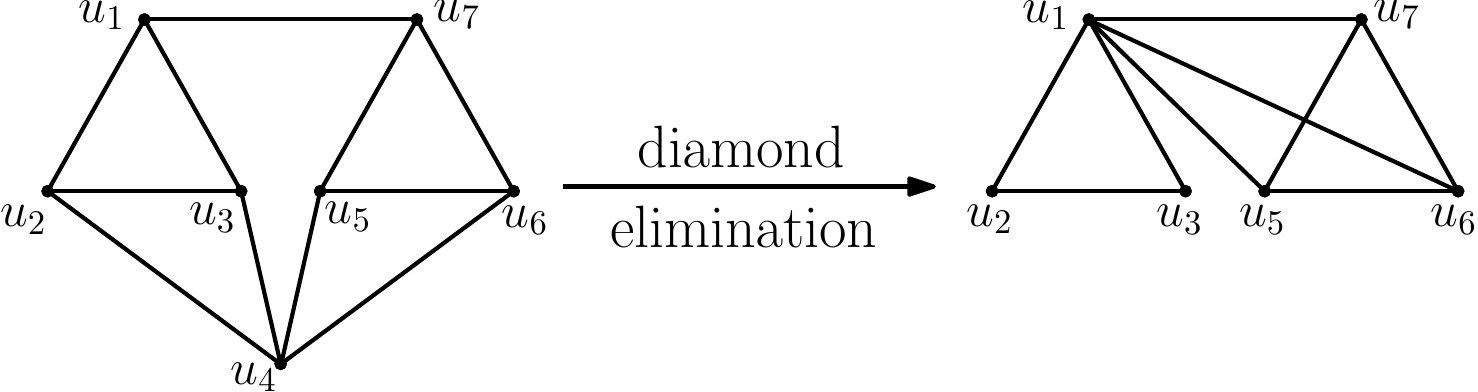}}
\caption{(a)~The diamond graph and~(b)~an example of a diamond elimination
of a $K_{4}$-free graph, which creates a new~$K_{4}$ on the vertices ${%
\{u_{1},u_{5},u_{6},u_{7}\}}$ in the resulting graph.}
\label{diamond-elimination-fig}
\end{figure}

\vspace{-0.2cm}

Suppose now that a graph $G$ has a pair of siblings $u$ and $v$ and assume
without loss of generality that~${N(u)\subseteq N(v)}$. Then, we can extend
any proper~$3$-coloring of $G-\{u\}$ (if such exists) to a proper~$3$%
-coloring of $G$ by assigning to $u$ the same color as $v$. Therefore, we
can remove vertex $u$ from $G$, as the next reduction rule states, and the
resulting graph $G-\{u\}$ is $3$-colorable if and only if $G$ is $3$%
-colorable.

\begin{reductionRule}[siblings elimination]
\label{rule2}Let ${G=(V,E)}$ be a $K_{4}$-free graph and ${u,v\in V}$, such
that ${N(u)\subseteq N(v)}$. Then remove $u$ from $G$.
\end{reductionRule}

\begin{definition}
\label{irreducible-def}Let $G=(V,E)$ be a $K_{4}$-free graph. If neither
Reduction Rule~\ref{rule1} nor Reduction Rule~\ref{rule2} can be applied to $%
G$, then $G$ is \emph{irreducible}.
\end{definition}

Due to Definition~\ref{irreducible-def}, a $K_{4}$-free graph is irreducible
if and only if it is diamond-free and siblings-free. Given a $K_{4}$-free
graph $G$ with $n$ vertices, clearly we can iteratively execute Reduction
Rules~\ref{rule1} and~\ref{rule2} in time polynomial on $n$, until we either
find a $K_{4}$ or none of the Reduction Rules~\ref{rule1} and~\ref{rule2}
can be further applied. If we find a $K_{4}$, then clearly the initial graph 
$G$ is not $3$-colorable. Otherwise, we transform~$G$ in polynomial time
into an irreducible ($K_{4}$-free) graph $G^{\prime }$ of smaller or equal
size, such that $G^{\prime }$ is $3$-colorable if and only if $G$ is $3$%
-colorable.

Observe that during the application of these reduction rules to a graph~$G$,
neither the diameter nor the radius of~$G$ increase. Moreover, note that in
the irreducible graph~$G^{\prime }$, the neighborhood~$N_{G^{\prime }}(u)$
of every vertex $u$ in $G^{\prime }$ induces a graph with maximum degree at
most~$1$, since otherwise~$G^{\prime }$ would have a $K_{4}$ or a diamond as
an induced subgraph. That is, the subgraph of~$G^{\prime }$ induced by~$%
N_{G^{\prime }}(u)$ contains only isolated vertices and isolated edges.
Furthermore, if $G$ (and thus also~$G^{\prime}$) is connected and if $%
G^{\prime}$ has more than two vertices, then the minimum degree of $%
G^{\prime }$ is~$\delta (G^{\prime})\geq 2$, since $G^{\prime }$ is
siblings-free. All these facts are summarized in the next observation. In
the remainder of the article, we assume that any given graph $G$ is
irreducible.

\begin{observation}
\label{obs1}Let ${G=(V,E)}$ be a connected $K_{4}$-free graph and ${%
G^{\prime }=(V^{\prime },E^{\prime })}$ be the irreducible graph obtained
from~$G$. If $G^{\prime}$ has more than two vertices, then ${\delta
(G^{\prime })\geq 2}$, ${diam(G^{\prime })\leq diam(G)}$, ${rad(G^{\prime
})\leq rad(G)}$, and~${G^{\prime }}$ is $3$-colorable if and only if~$G$ is $%
3$-colorable. Moreover, for every~${u\in V^{\prime }}$, $N_{G^{\prime }}(u)$
induces in~$G^{\prime }$ a graph with maximum degree~$1$.
\end{observation}

\section{Algorithms for $3$-coloring on graphs with diameter~$2$\label%
{diameter-two-sec}}

In this section we present our results on graphs with diameter~$2$. In
particular, we provide in Section~\ref{sqrt-n-alg-diam-2-subsec} our
subexponential algorithm for $3$-coloring on such graphs. We then provide,
for every~$n$, an example of an irreducible and triangle-free graph~$G_{n}$
with~$\Theta (n)$ vertices and diameter~$2$, which is $3$-colorable, has
minimum dominating set of size $\Theta (\sqrt{n})$, and its minimum degree
is $\delta(G_n) = \Theta (\sqrt{n})$. Furthermore, we define in Section~\ref%
{polynomial-L3-connected-diam-2-subsec} the notion of an articulation
neighborhood, and we provide our polynomial algorithm for irreducible graphs 
$G$ with diameter~$2$, which have at least one articulation neighborhood.

\subsection{An $2^{O(\protect\sqrt{n\log n})}$-time algorithm for any graph
with diameter~$2$\label{sqrt-n-alg-diam-2-subsec}}

We first provide in the next lemma a well known algorithm that decides the $%
3 $-coloring problem on an arbitrary graph $G$, using a dominating set (DS)
of $G$~\cite{NarayanaswamyIPL11}.

\begin{lemma}[\hspace{0.3pt}\protect\cite{NarayanaswamyIPL11}, the
DS-approach]
\label{3-coloring-dominating-set-lem}Let ${G=(V,E)}$ be a graph and $%
D\subseteq V$ be a dominating set of $G$. Then, the $3$-coloring problem can
be decided in $O^{\ast }(3^{|D|})$ time on $G$.
\end{lemma}

In the next theorem we use Lemma~\ref{3-coloring-dominating-set-lem} to
provide an improved $3$-coloring algorithm for the case of graphs with
diameter~$2$. The time complexity of this algorithm is parameterized on the
minimum degree $\delta $ of the given graph $G$, as well as on the fraction~$%
\frac{n}{\delta}$.

\begin{theorem}
\label{min-delta-n-over-delta-alg-diam-2-thm}Let ${G=(V,E)}$ be an
irreducible graph with $n$ vertices. Let ${diam(G)=2}$ and ${\delta }$ be
the minimum degree of $G$. Then, the $3$-coloring problem can be decided in $%
{2^{O(\min \{\delta ,\frac{n}{\delta }\log \delta \})}}$ time on $G$.
\end{theorem}

\begin{proof}
In an arbitrary graph $G$ with $n$ vertices and minimum degree $\delta $, it
is well known how to construct in polynomial time a dominating set $D$ with
cardinality $|D|\leq n\frac{1+\ln (\delta +1)}{\delta +1}$~\cite%
{AlonProbabilistic08} (see also~\cite{AlonTransversal90}). Therefore we can
decide by Lemma~\ref{3-coloring-dominating-set-lem} the $3$-coloring problem
on $G$ in time $O^{\ast }(3^{n\frac{1+\ln (\delta +1)}{\delta +1}})$. Note
by Observation~\ref{obs1} that $\delta \geq 2$, since $G$ is irreducible by
assumption. Therefore the latter running time is $2^{O(\frac{n}{\delta }\log
\delta )}$.

The DS-approach of Lemma~\ref{3-coloring-dominating-set-lem} applies to any
graph $G$. However, since $G$ has diameter~$2$ by assumption, we can design
a second algorithm for $3$-coloring of $G$ as follows. Consider a vertex $%
u\in V$ with minimum degree, i.e.~$\deg (u)=\delta $. Since $diam(G)=2$, it
follows that for every other vertex $u\in V$, either $d(u,v)=1$ or $d(u,v)=2$%
. Therefore $N(u)$ is a dominating set of $G$ with cardinality $\delta $.
Furthermore, in any possible $3$-coloring of $G$, every vertex of $N(u)$ can
be colored by one of two possible colors, since all vertices of $N(u)$ are
adjacent with $u$. We iterate now for all possible proper $2$-colorings of $%
N(u)$ (instead of all $3$-colorings of the dominating set in the proof of
Lemma~\ref{3-coloring-dominating-set-lem}). There are at most $2^{\delta }$
such colorings. Similarly to the algorithm of Lemma~\ref%
{3-coloring-dominating-set-lem}, for every such $2$-coloring of $N(u)$ we
solve in polynomial time the corresponding list $2$-coloring of $G-N[u]$.
Thus, considering at most all possible $2$-colorings of $N(u)$, we can
decide the $3$-coloring problem on $G$ in time $O^{\ast }(2^{\delta
})=2^{O(\delta )}$.

Summarizing, we can combine these two $3$-coloring algorithms for $G$,
obtaining an algorithm with time complexity ${2^{O(\min \{\delta ,\frac{n}{%
\delta }\log \delta \})}}$.\qed
\end{proof}

\medskip

The next corollary provides the first subexponential algorithm for the $3$%
-coloring problem on graphs with diameter~$2$. Its correctness follows now
by Theorem~\ref{min-delta-n-over-delta-alg-diam-2-thm}.

\begin{corollary}
\label{sqrt-n-alg-diam-2-cor}Let $G=(V,E)$ be an irreducible graph with $n$
vertices and let $diam(G)=2$. Then, the $3$-coloring problem can be decided
in $2^{O(\sqrt{n\log n})}$ time on $G$.
\end{corollary}

\begin{proof}
Let ${\delta =\delta (G)}$ be the minimum degree of $G$. If ${\delta \leq 
\sqrt{n\log n}}$, then the $3$-coloring problem can be decided in ${%
2^{O(\delta )}=2^{O(\sqrt{n\log n})}}$ time on $G$ by Theorem~\ref%
{min-delta-n-over-delta-alg-diam-2-thm}. Suppose now that ${\delta >\sqrt{%
n\log n}}$. Note that ${\log \delta <\log n}$, since ${\delta <n}$.
Therefore ${\frac{\log \delta }{\delta }<\frac{\log n}{\sqrt{n\log n}}=\sqrt{%
\frac{\log n}{n}}}$, and thus ${\frac{n}{\delta }\log \delta <n\sqrt{\frac{%
\log n}{n}}}$, i.e.~${\frac{n}{\delta }\log \delta <\sqrt{n\log n}}$.
Therefore the $3$-coloring problem can be decided in ${2^{O(\frac{n}{\delta }%
\log \delta )}=2^{O(\sqrt{n\log n})}}$ time on $G$ by Theorem~\ref%
{min-delta-n-over-delta-alg-diam-2-thm}.\qed
\end{proof}

\medskip

The time complexity of Corollary~\ref{sqrt-n-alg-diam-2-cor} is
asymptotically the same as the currently best known complexity for the graph
isomorphism (GI) problem~\cite{Babai83}. Thus, as the $3$-coloring problem
on graphs with diameter~$2$ has been neither proved to be polynomially
solvable nor to be NP-complete, it would be worth to investigate whether
this problem is polynomially reducible to/from the GI problem.

Given the statements of Lemma~\ref{3-coloring-dominating-set-lem} and
Theorem~\ref{min-delta-n-over-delta-alg-diam-2-thm}, a question that arises
naturally is whether the worst case complexity of the algorithm of Theorem~%
\ref{min-delta-n-over-delta-alg-diam-2-thm} is indeed $2^{O(\sqrt{n\log n})}$
(as given in Corollary~\ref{sqrt-n-alg-diam-2-cor}). That is, do there exist 
$3$-colorable irreducible graphs $G$ with $n$ vertices and $diam(G)=2$, such
that both $\delta (G)$ and the size of the minimum dominating set of $G$ are 
$\Theta (\sqrt{n\log n})$, or close to this value? We answer this question
to the affirmative, thus proving that, in the case of $3$-coloring of graphs
with diameter~$2$, our analysis of the DS-approach (cf.~Lemma~\ref%
{3-coloring-dominating-set-lem} and Theorem~\ref%
{min-delta-n-over-delta-alg-diam-2-thm}) is asymptotically almost tight. In
particular, we provide in the next theorem~for every $n$ an example of an
irreducible $3$-colorable graph $G_{n}$ with $\Theta (n)$ vertices and $%
diam(G_{n})=2$, such that both $\delta (G_{n})$ and the size of the minimum
dominating set of $G$ are $\Theta (\sqrt{n})$. In addition, each of these
graphs $G_{n}$ is triangle-free, as the next theorem states. The
construction of the graphs $G_{n}$ is based on a suitable and interesting
matrix arrangement of the vertices of $G_{n}$.

\begin{theorem}
\label{example-irreducible-sqrt-n-thm}Let ${n\geq 1}$. Then there exists an 
\emph{irreducible} and \emph{triangle-free} $3$\emph{-colorable} graph ${%
G_{n}=(V_{n},E_{n})}$ with $\Theta (n)$ vertices, where ${diam(G_{n})=2}$
and ${\delta (G}_{n}){=\Theta (\sqrt{n})}$. Furthermore, the size of the
minimum dominating set of $G_{n}$ is $\Theta (\sqrt{n})$.
\end{theorem}

\begin{proof}
We assume without loss of generality that ${n=4k^{2}+1}$ for some integer ${%
k\geq 1}$ and we construct a graph ${G_{n}=(V_{n},E_{n})}$ with $n$ vertices
(otherwise, if ${n\neq 4k^{2}+1}$ for any ${k\geq 1}$, we provide the same
construction of $G_{n}$ with ${4\lceil \frac{\sqrt{n}}{2}\rceil
^{2}+1=\Theta (n)}$ vertices). We arrange the first $n-1$ vertices of~$G_{n}$
(i.e. all vertices of $G_{n}$ except one of them) in a matrix of size~$%
2k\times 2k$. For simplicity of notation, we enumerate these $4k^{2}$
vertices in the usual fashion, i.e.~vertex $v_{i,j}$ is the vertex in the
intersection of row $i$ and column $j$, where~${1\leq i,j\leq 2k}$.
Furthermore, we denote the $(4k^{2}+1)$th vertex of $G_{n}$ by $v_{0}$. We
assign the~$3$~colors red, blue, green to the vertices of $G_{n}$ as
follows. The vertices ${\{v_{i,j}:1\leq i\leq k,\ 1\leq j\leq k\}}$ are
colored blue, the vertices ${\{v_{i,j}:k+1\leq i\leq 2k,\ 1\leq j\leq k\}}$
are colored green, the vertices ${\{v_{i,j}:1\leq i\leq 2k,\ k+1\leq j\leq
2k\}}$ are colored red, and vertex $v_{0}$ is colored green.

We add edges among vertices of $V_{n}$ as follows. First, vertex $v_{0}$ is
adjacent to exactly all vertices that are colored red in the above $3$%
-coloring of $G_{n}$. Then, the vertices of $j$th column $%
\{v_{1,j},v_{2,j},\ldots ,v_{2k,j}\}$ and the vertices of the $(2k+1-j)$th
column $\{v_{1,2k+1-j},v_{2,2k+1-j},\ldots ,v_{2k,2k+1-j}\}$ build a
complete bipartite graph, without the edges $\{v_{i,j}v_{i,2k+1-j}:1\leq
i\leq 2k\}$, i.e.~without the edges between vertices of the same row.
Furthermore, the vertices of $i$th row $\{v_{i,1},v_{i,2},\ldots ,v_{i,2k}\}$
and the vertices of the $(2k+1-i)$th row $\{v_{2k+1-i,1},v_{2k+1-i,2},\ldots
,v_{2k+1-i,2k}\}$ build a complete bipartite graph, without the edges $%
\{v_{i,j}v_{2k+1-i,j}:1\leq j\leq k\}\cup \{v_{i,j}v_{2k+1-i,\ell }:k+1\leq
j,\ell \leq 2k\}$. That is, there are no edges between vertices of the same
column and no edges between vertices colored red in the above $3$-coloring
of~$G_{n}$. Note also that there are no edges between vertices colored blue
(resp.~green, red), and thus this coloring is a proper $3$-coloring of $%
G_{n} $. The $2k\times 2k$ matrix arrangement of the vertices of $G_{n}$ is
illustrated in Figure~\ref{example-irreducible-fig-1}. In this figure, the
three color classes are illustrated by different shades of gray.
Furthermore, the edges of $G_{n}$ between different rows and between
different columns in this matrix arrangement are illustrated in Figures~\ref%
{example-irreducible-fig-2} and~\ref{example-irreducible-fig-3},
respectively.

\begin{figure}[h!tb]
\centering 
\subfigure[]{ \label{example-irreducible-fig-1}
\includegraphics[scale=0.363]{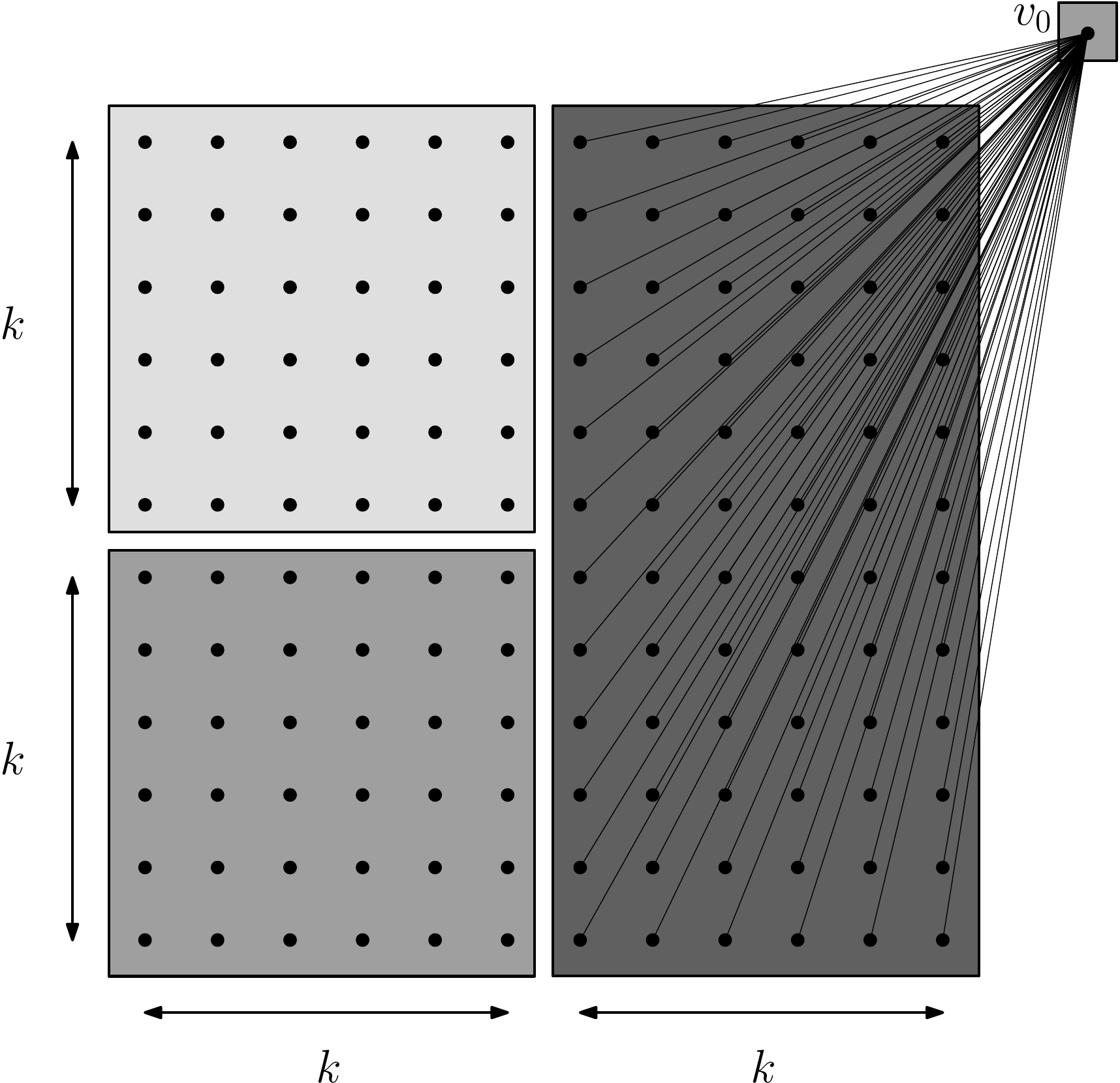}} 
\subfigure[]{ \label{example-irreducible-fig-2}
\includegraphics[scale=0.363]{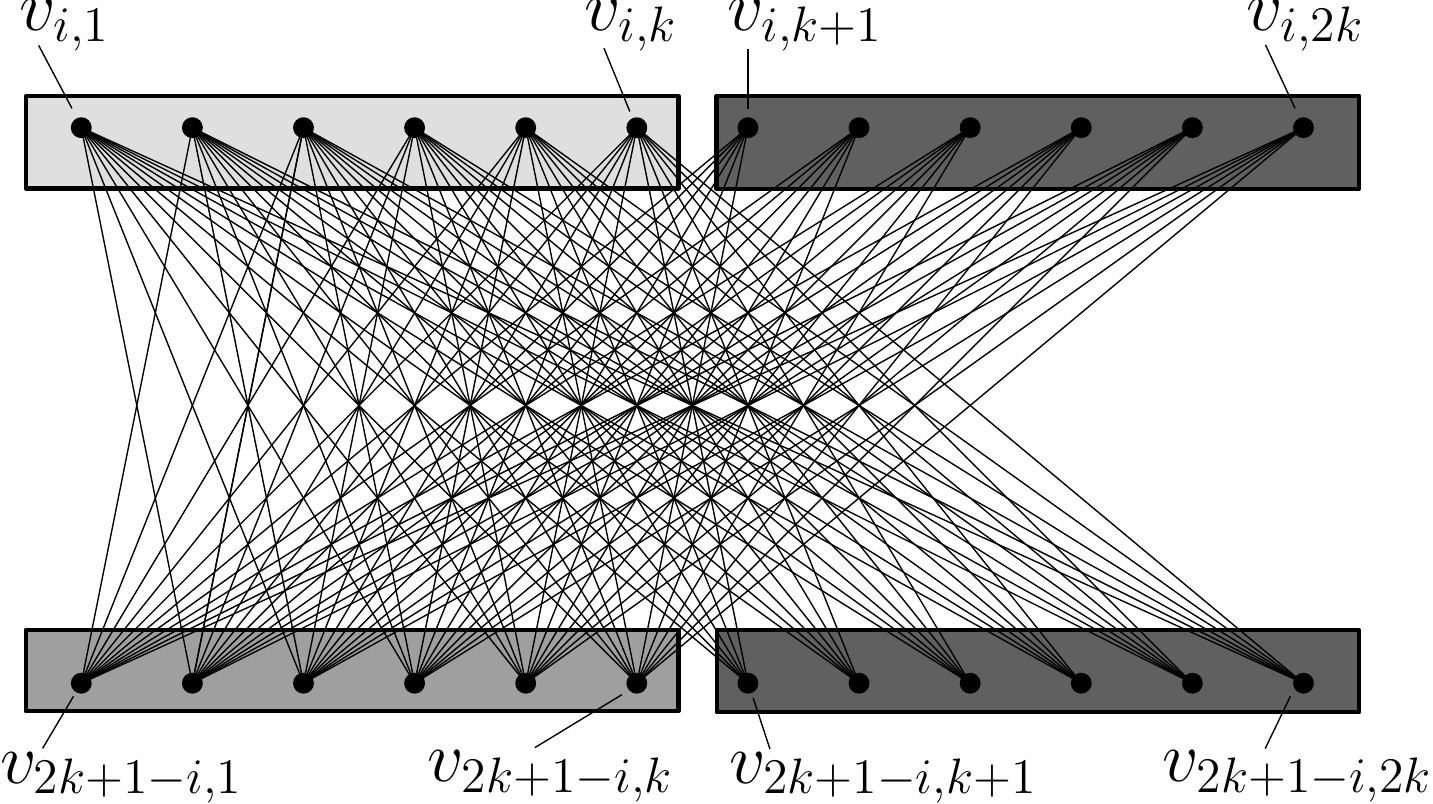}} \hspace{0.12cm}
\subfigure[]{ \label{example-irreducible-fig-3}
\includegraphics[scale=0.363]{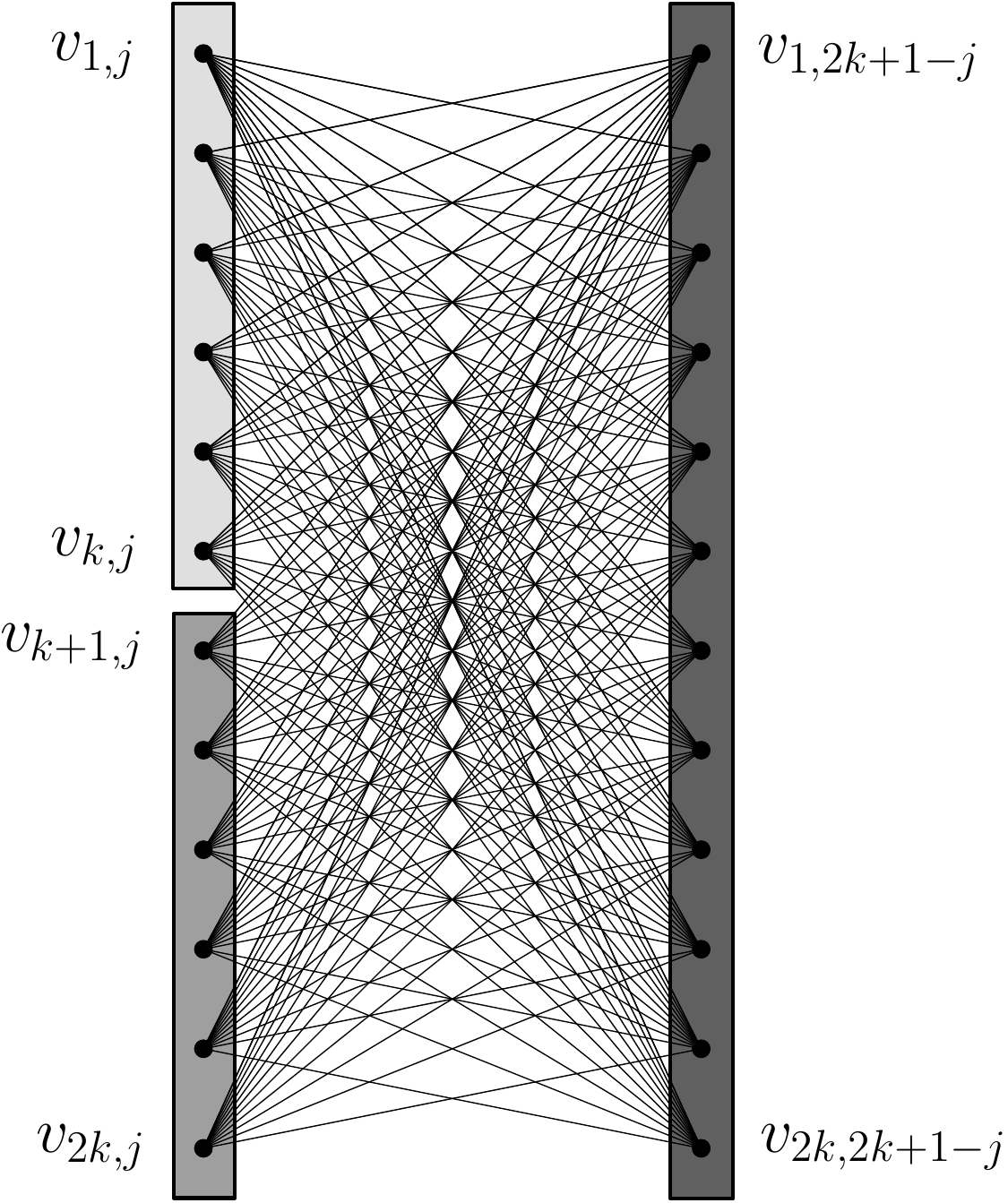}}
\caption{(a)~The ${2k\times 2k}$ matrix arrangement of the vertices of $%
G_{n} $, where~${n=4k^{2} + 1}$, (b)~the edges of~$G_{n}$ between the~$i$th and
the~${(2k+1-i)}$th rows, and~(c)~the edges of~$G_{n}$ between the~$j$th and
the ${(2k+1-j)}$th columns of this matrix arrangement. The three color
classes are illustrated by different shades of gray.}
\label{example-irreducible-fig}
\end{figure}

It is easy to see by the construction of $G_{n}$ that $3k\leq \deg
(v_{i,j})\leq 4k-2$ for every vertex $v_{i,j}\in V_{n}\setminus \{v_{0}\}$
and that $\deg (v_{0})=2k^{2}$. In particular, $\deg (v_{i,j})=3k$ (resp.~$%
\deg (v_{i,j})=4k-2$) for every vertex $v_{i,j}$ that has been colored red
(resp.~blue or green) in the above coloring of $G_{n}$. Therefore, since $%
k=\Theta (\sqrt{n})$, it follows that $\deg (v_{i,j})=\Theta (\sqrt{n})$ for
every $v_{i,j}\in V_{n}\setminus \{v_{0}\}$ and that $\deg (v_{0})=\frac{n-1%
}{2}$. That is, the minimum degree is $\delta (G_{n})=\Theta (\sqrt{n})$.

Note now that the set $\{v_{0}\}\cup \{v_{1,1},v_{2,1},\ldots ,v_{2k,1}\}$
of vertices is a dominating set of $G$ with $2k+1=\Theta (\sqrt{n})$
vertices, since for every $i=1,2,\ldots ,2k$, vertex $v_{i,1}$ is adjacent
to all vertices of the $(2k+1-i)$th row of the matrix, except vertex $%
v_{2k+1-i,1}$. Suppose that there exists a dominating set $D\subseteq V_{n}$
with cardinality~$o(\sqrt{n})$, i.e. $|D|=\frac{\sqrt{n}}{f(n)}$, where $%
f(n)=\omega (1)$. Denote $D^{\prime }=D\cup \{v_{0}\}$; note that $%
|D^{\prime }|\leq \frac{\sqrt{n}}{f(n)}+1$. Then, since $\deg
(v_{i,j})=\Theta (\sqrt{n})$ for every vertex $v_{i,j}\in V_{n}\setminus
\{v_{0}\}$, the vertices of $D^{\prime }$ can be adjacent to at most 
\begin{equation*}
\frac{n-1}{2}+|D^{\prime }\setminus \{v_{0}\}|\cdot \Theta (\sqrt{n})\leq 
\frac{n-1}{2}+\frac{\sqrt{n}}{f(n)}\cdot \Theta (\sqrt{n})=\frac{n-1}{2}+%
\frac{\Theta (n)}{f(n)}
\end{equation*}%
vertices of $V_{n}\setminus D^{\prime }$ in total. However $|V_{n}\setminus
D^{\prime }|\geq n-\frac{\sqrt{n}}{f(n)}-1\geq n-\sqrt{n}-1$ and $\frac{n-1}{%
2}+\frac{\Theta (n)}{f(n)}$ is asymptotically smaller than $n-\sqrt{n}-1$,
since $f(n)=\omega (1)$. That is, the vertices of $D^{\prime }$ (and thus
also of $D$) can not dominate all vertices of $V_{n}\setminus D^{\prime
}\subseteq V_{n}\setminus D$. Thus $D$ is not a dominating set, which is a
contradiction. Therefore the size of a minimum dominating set of $G_{n}$ is~$%
\Theta (\sqrt{n})$.

Now we will prove that $diam(G_{n})=2$. First note that for every vertex $%
v_{i,j}\in V_{n}\setminus \{v_{0}\}$ we have $d(v_{0},v_{i,j})\leq 2$.
Indeed, if $v_{i,j}$ is colored red, then $v_{i,j}$ is adjacent to $v_{0}$;
otherwise $v_{0}$ and $v_{i,j}$ have vertex $v_{2k+1-i,2k}$ as common
neighbor. Now consider two arbitrary vertices $v_{i,j}$ and $v_{p,q}$, where 
$(i,j)\neq (p,q)$. Since any two red vertices have $v_{0}$ as common
neighbor (and thus they are at distance $2$ from each other), assume that at
most one of the vertices $\{v_{i,j},v_{p,q}\}$ is colored red. We will prove
that $d(v_{i,j},v_{p,q})\leq 2$. If $p=i$, then $v_{i,j}$ and $v_{p,q}$ lie
both in the $i$th row of the matrix. Therefore $v_{i,j}$ and $v_{p,q}$ have
all vertices of $\{v_{2k+1-i,1},v_{2k+1-i,2},\ldots ,v_{2k+1-i,k}\}\setminus
\{v_{2k+1-i,j},v_{2k+1-i,q}\}$ as their common neighbors, and thus $%
d(v_{i,j},v_{p,q})=2$. If $q=j$, then $v_{i,j}$ and $v_{p,q}$ lie both in
the $j$th column of the matrix. Therefore $v_{i,j}$ and $v_{p,q}$ have all
vertices of $\{v_{1,2k+1-j},v_{2,2k+1-j},\ldots ,v_{2k,2k+1-j}\}\setminus
\{v_{i,2k+1-j},v_{p,2k+1-j}\}$ as their common neighbors, and thus $%
d(v_{i,j},v_{p,q})=2$. Suppose that $p\neq i$ and $q\neq j$. If $p=2k+1-i$
or $q=2k+1-j$, then $v_{i,j}v_{p,q}\in E_{n}$ by the construction of $G_{n}$%
, and thus $d(v_{i,j},v_{p,q})=1$. Suppose now that also $p\neq 2k+1-i$ and $%
q\neq 2k+1-j$. Since at most one of the vertices $\{v_{i,j},v_{p,q}\}$ is
colored red, we may assume without loss of generality that $v_{i,j}$ is blue
or green, i.e. $j\leq k$. Then there exists the path $%
(v_{i,j},v_{2k+1-i,2k+1-q},v_{p,q})$ in $G_{n}$, and thus $%
d(v_{i,j},v_{p,q})=2$. Summarizing, $d(u,v)\leq 2$ for every pair of
vertices $u$ and $v$ in $G_{n}$, and thus $diam(G_{n})=2$.

Now we will prove that $G_{n}$ is a triangle-free graph. First note that $%
v_{0}$ can not belong to any possible triangle in $G_{n}$, since its
neighbors are by construction an independent set. Suppose now that the
vertices $v_{i,j},v_{p,q},v_{r,s}$ induce a triangle in $G_{n}$. Since the
above coloring of the vertices of $V_{n}$ is a proper $3$-coloring of $G_{n}$%
, we may assume without loss of generality that $v_{i,j}$ is colored blue, $%
v_{p,q}$ is colored green, and $v_{r,s}$ is colored red in the above
coloring of $G_{n}$. That is, $i,j\in \{1,2,\ldots ,k\}$, $p\in
\{k+1,k+2,\ldots ,2k\}$, $q\in \{1,2,\ldots ,k\}$, $r\in \{1,2,\ldots ,2k\}$%
, and $s\in \{k+1,k+2,\ldots ,2k\}$. Thus, since $v_{i,j}v_{p,q}\in E_{n}$,
it follows by the construction of $G_{n}$ that $p=2k+1-i$. Furthermore,
since $v_{p,q}v_{r,s}\in E_{n}$, it follows that $p=2k+1-r$ or $q=2k+1-s$.
If $p=2k+1-r$, then $r=i$ (since also $p=2k+1-i$). Therefore $v_{i,j}$ and $%
v_{r,s}$ lie both in the $i$th row of the matrix, which is a contradiction,
since we assumed that $v_{i,j}v_{r,s}\in E_{n}$. Therefore $q=2k+1-s$.
Finally, since $v_{i,j}v_{r,s}\in E_{n}$, it follows that $r=2k+1-i$ or $%
s=2k+1-j$. If $r=2k+1-i$, then $r=p$ (since also $p=2k+1-i$). Therefore $%
v_{p,q}$ and $v_{r,s}$ lie both in the $p$th row of the matrix, which is a
contradiction, since we assumed that $v_{p,q}v_{r,s}\in E_{n}$. Therefore $%
s=2k+1-j$. Thus, since also $q=2k+1-s$, it follows that $q=j$, i.e.~$v_{i,j}$
and $v_{p,q}$ lie both in the $j$th column of the matrix. This is again a
contradiction, since we assumed that $v_{i,j}v_{p,q}\in E_{n}$. Therefore,
no three vertices of $G_{n}$ induce a triangle, i.e.~$G_{n}$ is
triangle-free.

Note now that $G_{n}$ is diamond-free, since it is also triangle-free, and
thus the Reduction Rule~\ref{rule1} does not apply on $G_{n}$. Furthermore,
it is easy to check that there exist no pair $v_{i,j}$ and $v_{p,q}$ of
vertices such that $N(v_{i,j})\subseteq N(v_{p,q})$, i.e.~$G_{n}$ is also
siblings-free, and thus also the Reduction Rule~\ref{rule2} does not apply
on $G_{n}$. Therefore $G_{n}$ is irreducible. This completes the proof of
the theorem.\qed
\end{proof}

\subsection{A tractable subclass of graphs with diameter~$2$\label%
{polynomial-L3-connected-diam-2-subsec}}

In this section we present a subclass of graphs with diameter~$2$, which
admits an efficient algorithm for $3$-coloring. We first introduce the
definition of an \emph{articulation neighborhood} in a graph.

\begin{definition}
Let $G=(V,E)$ be a graph and let $v\in V$. If $G-N[v]$ is disconnected, then 
$N[v]$ is an \emph{articulation neighborhood} in $G$.
\end{definition}

We prove in Theorem~\ref{polynomial-L3-connected-diam-2-thm} that, given an
irreducible graph with~$diam(G)=2$, which has at least one articulation
neighborhood, Algorithm~\ref{polynomial-L3-connected-diam-2-alg} decides~$3$%
-coloring on~$G$ in polynomial time. Note here that there exist instances of~%
$K_{4}$-free graphs~$G$ with diameter~$2$, for $G$ has no articulation
neighborhood, but in the irreducible graph~$G^{\prime}$ obtained by~$G$ (by
iteratively applying the Reduction Rules~\ref{rule1} and~\ref{rule2}), $%
G^{\prime}-N_{G^{\prime}}[v]$ becomes disconnected for some vertex~$v$ of~$%
G^{\prime}$. That is, $G^{\prime}$ may have an articulation neighborhood,
although $G$ has none. Therefore, if we provide as input to Algorithm~\ref%
{polynomial-L3-connected-diam-2-alg} the irreducible graph~$G^{\prime}$
instead of~$G$, this algorithm decides in polynomial time the~$3$-coloring
problem on~$G^{\prime}$ (and thus also on~$G$).

\begin{algorithm}[h!tb]
\caption{$3$-coloring of an irreducible graph $G$ with diameter-$2$ and at least one articulation neighborhood} \label{polynomial-L3-connected-diam-2-alg}
\begin{algorithmic}[1]
\REQUIRE{An irreducible graph $G=(V,E)$ with $diam(G)=2$ and a vertex $v_{0} \in V$, such that $G-N[v_{0}]$ is disconnected}
\ENSURE{A proper $3$-coloring of $G$, or the announcement that $G$ is not $3$-colorable}

\medskip

\STATE{Compute the connected components $C_{1},C_{2},\ldots,C_{k}$ of $G-N[v_{0}]$} \label{alg-3}

\STATE{Color $v_{0}$ red} \label{alg-4}

\medskip

\FOR{$i=1$ to $k$} \label{alg-5}
     \IF{$C_{i}$ is bipartite} \label{alg-6}
          \STATE{merge the two color classes of $C_{i}$ into two vertices $u_{i}$ and $v_{i}$, respectively\ \ } \label{alg-7}
          \COMMENT{$u_{i}$ and $v_{i}$ are adjacent}
     \ELSE \label{alg-8}
          \RETURN{``$G$ is not $3$-colorable''} \label{alg-9}
     \ENDIF \label{alg-10}
\ENDFOR \label{alg-11}

\medskip

\STATE{$G^{\prime} \leftarrow G$} \label{alg-12}

\WHILE{Reduction Rule~\ref{rule1} or~\ref{rule2} can be applied to $G^{\prime}$, or $G^{\prime}$ contains no induced $K_{4}$} \label{alg-13}
     \STATE{Apply Rule~\ref{rule1} or~\ref{rule2} to $G^{\prime}$} \label{alg-14}
\ENDWHILE \label{alg-15}

\medskip

\IF{$G^{\prime}$ contains an induced $K_{4}$} \label{alg-16}
     \RETURN{``$G$ is not $3$-colorable''} \label{alg-17}
\ENDIF \label{alg-18}

\medskip

\STATE{$N_{0} \leftarrow \{u\in N_{G^{\prime}}(v_{0}) : |N(u) \cap N_{G^{\prime}}(v_{0})|=0\}$; \ \ \ $N_{1} \leftarrow N_{G^{\prime}}(v_{0}) \setminus N_{0}$} \label{alg-19}

\medskip

\IF{$N_{0} \neq \emptyset$} \label{alg-20}
     \STATE{Let $N_{0} = \{w_{1},w_{2},\ldots,w_{p}\}$} \label{alg-21}
\ENDIF \label{alg-22}

\IF{$N_{1} \neq \emptyset$} \label{alg-23}
     \STATE{Let $N_{1} = \{z_{1},z_{2},\ldots,z_{2q}\}$, where $z_{2i-1}z_{2i}\in E^{\prime}$ for every $i=1,2,\ldots,q$} \label{alg-24}
\ENDIF \label{alg-25}

\medskip

\IF{$G^{\prime} - N_{G^{\prime}}[v_{0}] = \emptyset$} \label{alg-26}
     
     \STATE{Compute a $2$-coloring of $N_{G^{\prime}}(v_{0})$ with colors blue and green\\} \label{alg-27}
     \COMMENT{this $2$-coloring, together with the red color of~$v_{0}$ (cf.~line~\ref{alg-3}) constitutes a~$3$-coloring of $G^{\prime}$}
     
     \vspace{0.1cm}
     
\ELSE[$G^{\prime} - N_{G^{\prime}}{[}v_{0}{]} \neq \emptyset$] \label{alg-28}
     
     \vspace{0.1cm}
     
     \STATE{Compute the connected components $C_{1}^{\prime},C_{2}^{\prime},\ldots,C_{k^{\prime}}^{\prime}$ of $G^{\prime}-N_{G^{\prime}}[v_{0}]$} \label{alg-29}
     
     \STATE{Let $C_{i}^{\prime} = \{u_{i}^{\prime}, v_{i}^{\prime}\}$ for every $i=1,2,\ldots,k^{\prime}$} \label{alg-30}
     
     \medskip
     
     \IF{$N_{1} = \emptyset$} \label{alg-31}

          \STATE{Color all vertices of $N_{G^{\prime}}(v_{0})$ green} \label{alg-32}

          \STATE{Color all vertices $u_{i}^{\prime}$ blue and all vertices $v_{i}^{\prime}$ red, for every $i=1,2,\ldots,k^{\prime}$} \label{alg-33}
          
          \vspace{0.1cm}
          
     \ELSE[$N_{1} \neq \emptyset$] \label{alg-34}
          
          \vspace{0.1cm}
          
          \STATE{$\phi \leftarrow \emptyset$; \ Define boolean variables $x_{i}$ and $y_{j}$, where $1\leq i \leq k^{\prime}$, $1\leq j \leq q$} \label{alg-35}
          
          \STATE{\textbf{if} \ $u_{i}^{\prime} z_{2j-1} \in E^{\prime}$ \ \textbf{then} \ $\phi \leftarrow \phi \wedge (x_{i} \vee \overline{y_{j}})$ 
                                                                        \ \textbf{else} \ $\phi \leftarrow \phi \wedge (x_{i} \vee y_{j})$} \label{alg-36}
                                                                      
          \STATE{\textbf{if} \ $v_{i}^{\prime} z_{2j-1} \in E^{\prime}$ \ \textbf{then} \ $\phi \leftarrow \phi \wedge (\overline{x_{i}} \vee \overline{y_{j}})$ 
                                                                        \ \textbf{else} \ $\phi \leftarrow \phi \wedge (\overline{x_{i}} \vee y_{j})$} \label{alg-37}
          
          \medskip
          
          \IF{$\phi$ is not satisfiable} \label{alg-38}
          
               \RETURN{``$G$ is not $3$-colorable''} \label{alg-39}
               
          \ELSE \label{alg-40}
               
               \STATE{Compute a truth assignment $\tau$ that satisfies the formula $\phi$} \label{alg-41a}
               
               \FOR{$i=1$ to $k^{\prime}$} \label{alg-41b}
                    \STATE{\textbf{if} \ $x_{i}=1$ in $\tau$ \ \textbf{then} \ color $u_{i}^{\prime}$ red and $v_{i}^{\prime}$ blue \ \textbf{else} \ color $u_{i}^{\prime}$ blue and $v_{i}^{\prime}$ red} \label{alg-42}
               \ENDFOR \label{alg-43}
               
               \FOR{$j=1$ to $q$} \label{alg-44}
                    \STATE{\textbf{if} \ $y_{i}=1$ in $\tau$ \ \textbf{then} \ color $z_{2j-1}$ green and $z_{2j}$ blue \ \textbf{else} \ color $z_{2j-1}$ blue and $z_{2j}$ green} \label{alg-45}
               \ENDFOR \label{alg-46}
               
          \ENDIF \label{alg-47}
          
     \ENDIF \label{alg-48}
     
\ENDIF \label{alg-49}

\medskip

\STATE{From the computed $3$-coloring of $G^{\prime}$, compute a $3$-coloring $\chi$ of $G$ by iteratively undoing Reduction Rules~\ref{rule1} and~\ref{rule2}}\label{alg-50}

\medskip

\RETURN{the proper $3$-coloring $\chi$ of $G$} \label{alg-51}
\end{algorithmic}
\end{algorithm}

\begin{theorem}
\label{polynomial-L3-connected-diam-2-thm} Let ${G=(V,E)}$ be irreducible
with $n $ vertices and ${diam(G)=2}$. If $G$ has at least one articulation
neighborhood, then Algorithm~\ref{polynomial-L3-connected-diam-2-alg}
decides $3$-coloring on $G$ in time polynomial in~$n$.
\end{theorem}

\begin{proof}
Let~$v_{0}\in V$ such that~${G-N[v_{0}]}$ is disconnected and let $%
C_{1},C_{2},\ldots ,C_{k}$ be the connected components of $G-N[v_{0}]$
(cf.~line~\ref{alg-3} of Algorithm~\ref{polynomial-L3-connected-diam-2-alg}%
). In order to compute a proper $3$-coloring of~$G$, the algorithm assigns
without loss of generality the color red to vertex~$v_{0}$ (cf.~line~\ref%
{alg-4}). Suppose that at least one component~$C_{i}$, $1\leq i\leq k$, is
trivial, i.e.~$C_{i}$ is an isolated vertex~$u$. Then, since~${u\in
V\setminus N[v_{0}]}$ and~$u$ has no adjacent vertices in~$V\setminus
N[v_{0}]$, it follows that~$N(u)\subseteq N(v_{0})$, i.e.~$u$ and~$v_{0}$
are siblings in~$G$. This is a contradiction, since~$G$ is an irreducible
graph. Therefore, every connected component~$C_{i}$ of~${G-N[v_{0}]}$ is
non-trivial, i.e.~it contains at least one edge. Thus, in any proper~$3$%
-coloring of~$G$ (if such exists), there exists at least one vertex of every
connected component~$C_{i}$ of~${G-N[v_{0}]}$ that is colored not red,
i.e.~with a different color than $v_{0}$.

Suppose now that $G$ is $3$-colorable, and let $\chi $ be a proper $3$%
-coloring of $G$. Assume without loss of generality that $\chi $ uses the
colors red, blue, and green (recall that $v_{0}$ is already colored red in
line~\ref{alg-4} of the algorithm). Let $C_{i},C_{j}$ be an arbitrary pair
of connected components of $G-N[v_{0}]$; such a pair always exists, as $%
G-N[v_{0}]$ is assumed to be disconnected. Let $u$ and $v$ be two arbitrary
vertices of~$C_{i}$ and of~$C_{j}$, respectively, such that both $u$ and $v$
are not colored red in $\chi $; such vertices~$u$ and~$v$ always exist, as
we observed above. Then, since $diam(G)=2$, there exists at least one common
neighbor $a$ of $u$ and $v$, where $a\in N(v_{0})$. That is, $a\in
N(v_{0})\cap N(u)\cap N(v)$. Therefore, since $v_{0}$ is colored red in $%
\chi $ and $u,v$ are not colored red in $\chi $, it follows that both $u$
and $v$ are colored by the same color in $\chi $. Assume without loss of
generality that both $u$ and $v$ are colored blue in~$\chi$. Thus, since $u$
and $v$ have been chosen arbitrarily under the single assumption that they
are not colored red in~$\chi $, it follows that for every vertex $u$ of $%
G-N[v_{0}]$, $u$ is either colored red or blue in~$\chi $. Therefore $%
G-N[v_{0}]$ is bipartite, i.e.~each of the components $C_{1},C_{2},\ldots
,C_{k}$ is bipartite. Thus, Algorithm~\ref%
{polynomial-L3-connected-diam-2-alg} correctly returns in line~\ref{alg-9}
that $G$ is not $3$-colorable if at least one component $C_{i}$ of~${%
G-N[v_{0}]}$ is not bipartite.

Since all connected components $C_{1},C_{2},\ldots ,C_{k}$ of $G-N[v_{0}]$
are bipartite, Algorithm~\ref{polynomial-L3-connected-diam-2-alg} merges for
every $i\in \{1,2,\ldots ,k\}$ the vertices of each of the two color-classes
of $C_{i}$ into one vertex (cf.~lines~\ref{alg-5}-\ref{alg-7}). That is, for
every $i\in \{1,2,\ldots ,k\}$, the component $C_{i}$ is replaced by exactly
two adjacent vertices~$\{u_{i},v_{i}\}$, where the one color class of $C_{i}$
is merged into $u_{i}$ and the other one is merged into~$v_{i}$ (cf.~line~%
\ref{alg-7}). Then Algorithm~\ref{polynomial-L3-connected-diam-2-alg}
iteratively applies the Reduction Rules~\ref{rule1} and~\ref{rule2}, until
none of them can be further applied, or until it detects an induced $K_{4}$.
Denote the resulting graph by $G^{\prime }=(V^{\prime },E^{\prime })$,
cf.~lines~\ref{alg-12}-\ref{alg-14} of the algorithm. If $G^{\prime }$
contains an induced $K_{4}$, then $G^{\prime }$ is clearly not $3$%
-colorable, and thus the initial graph $G$ is also not $3$-colorable.
Therefore, in this case, Algorithm~\ref{polynomial-L3-connected-diam-2-alg}
correctly returns in line~\ref{alg-17} that $G$ is not $3$-colorable.
Otherwise, if $G^{\prime }$ does not contain any induced $K_{4}$, then $%
G^{\prime }$ is irreducible by Definition~\ref{irreducible-def}. Note that $%
diam(G^{\prime })=2$ and that $G^{\prime }$ is $3$-colorable, since $%
diam(G)=2$ and $G$ is assumed to be~$3$-colorable as well. Denote by $\chi
^{\prime }$ the $3$-coloring of $G^{\prime }$ that is induced by the $3$%
-coloring $\chi $ of $G$. Then, similarly to $\chi $, every vertex of $%
G^{\prime }-N_{G^{\prime }}[v_{0}]$ is either colored red or blue in $\chi
^{\prime }$.

Note by Observation~\ref{obs1} that $N_{G^{\prime }}(v_{0})$ induces in~$%
G^{\prime }$ a graph with maximum degree~$1$, since $G^{\prime }$ is
irreducible. That is, the vertices of $N_{G^{\prime }}(v_{0})$ can be
partitioned into the sets $N_{0}$ and $N_{1}$ that include the vertices of $%
G^{\prime }[N_{G^{\prime }}(v_{0})]$ with degree $0$ and with degree $1$,
respectively (cf.~line~\ref{alg-19} of Algorithm~\ref%
{polynomial-L3-connected-diam-2-alg}). That is, $N_{0}$ contains the
isolated vertices of $G^{\prime }[N_{G^{\prime }}(v_{0})]$ and $N_{1}$
contains the vertices of $N_{G^{\prime }}(v_{0})$ that are paired in edges
in $G^{\prime }[N_{G^{\prime }}(v_{0})]$. Let $N_{0}=\{w_{1},w_{2},\ldots
,w_{p}\}$ and $N_{1}=\{z_{1},z_{2},\ldots ,z_{2q}\}$, where $%
z_{2i-1}z_{2i}\in E$ for every~$i=1,2,\ldots ,q$ (cf.~lines~\ref{alg-20}-\ref%
{alg-24} of the algorithm). Furthermore, note that $N_{G^{\prime }}(v_{0})$
is bipartite.

After the execution of the latter applications of the Reduction Rules~\ref%
{rule1} and~\ref{rule2} to $G$ (cf.~lines~\ref{alg-13}-\ref{alg-14} of the
algorithm), some of the vertices of $G-N_{G}[v_{0}]$ may not appear any more
in $G^{\prime }-N_{G^{\prime }}[v_{0}]$. If~$G^{\prime }-N_{G^{\prime
}}[v_{0}]=\emptyset$, then $v_{0}$ is adjacent to all other vertices of $%
G^{\prime }$. Therefore, since $N_{G^{\prime }}(v_{0})$ is bipartite,
Algorithm~\ref{polynomial-L3-connected-diam-2-alg} correctly computes a $2$%
-coloring of $N_{G^{\prime }}(v_{0})$ with colors blue and green (cf.~line~%
\ref{alg-27}). This $2$-coloring of $N_{G^{\prime }}(v_{0})$, together with
the red color of $v_{0}$, constitutes a $3$-coloring of $G^{\prime }$.
Suppose now that $G^{\prime }-N_{G^{\prime }}[v_{0}]\neq \emptyset $. Denote
the connected components of $G^{\prime }-N_{G^{\prime }}[v_{0}]$ by $%
C_{1}^{\prime },C_{2}^{\prime },\ldots ,C_{k^{\prime }}^{\prime }$, where ${%
1\leq k^{\prime }\leq k}$ (cf.~line~\ref{alg-29}). Due to the above merging
operations to the connected components~$C_{1},C_{2},\ldots ,C_{k}$ of $%
G-N[v_{0}]$, every connected component $C_{i}^{\prime }$ of $G^{\prime
}-N_{G^{\prime }}[v_{0}]$ has exactly two vertices, where $1\leq i\leq
k^{\prime }$. Denote the vertices of $C_{i}^{\prime }$ by $\{u_{i}^{\prime
},v_{i}^{\prime }\}$, where $1\leq i\leq k^{\prime }$ (cf.~line~\ref{alg-30}%
).

Suppose that $N_{1}=\emptyset $, i.e.~if $q=0$, and thus $N_{G^{\prime
}}(v_{0})=N_{0}$. Then we can construct a proper $3$-coloring of $G^{\prime
} $ as follows. We first color all vertices of $N_{G^{\prime }}(v_{0})$
green. Then, for every $i=1,2,\ldots ,k^{\prime }$, we color $u_{i}$ blue
and $v_{i} $ red (cf.~lines~\ref{alg-31}-\ref{alg-33}). Note that this
coloring is a proper $3$-coloring of $G^{\prime }$.

Suppose now that $N_{1}\neq \emptyset $, i.e.~$q\neq 0$ (cf.~line~\ref%
{alg-34}). We prove that $|N_{G^{\prime }}(u_{i})\cap
\{z_{2j-1},z_{2j}\}|=|N_{G^{\prime }}(v_{i})\cap \{z_{2j-1},z_{2j}\}|=1$ for
every $i\in \{1,2,\ldots ,k^{\prime }\}$ and every $j\in \{1,2,\ldots ,q\}$.
Consider a vertex~$u_{i}\in V^{\prime }\setminus N_{G^{\prime }}[v_{0}]$,
where $1\leq i\leq k^{\prime }$. Suppose that there exists an index $j\in
\{1,2,\ldots ,q\}$ such that both $z_{2j-1},z_{2j}\in N_{G^{\prime }}(u_{i})$%
. Then the vertices $\{u_{i},v_{0},z_{2j-1},z_{2j}\}$ induce a diamond graph
in $G^{\prime }$ (with $u_{i}v_{0}\notin E^{\prime }$). This is a
contradiction, since $G^{\prime }$ is irreducible. Therefore ${|N_{G^{\prime
}}(u_{i})\cap \{z_{2j-1},z_{2j}\}|\leq 1}$ for every $j\in \{1,2,\ldots ,q\}$%
. Similarly we can prove that also $|N_{G^{\prime }}(v_{i})\cap
\{z_{2j-1},z_{2j}\}|\leq 1$ for every $j\in \{1,2,\ldots ,q\}$. Suppose now
that there exists an index $j\in \{1,2,\ldots ,q\}$ such that both~$%
z_{2j-1},z_{2j}\notin N_{G^{\prime }}(u_{i})$. Then, since $diam(G^{\prime
})=2$, there must exist paths $(u_{i},a,z_{2j-1})$ and $(u_{i},a^{\prime
},z_{2j})$ in $G^{\prime }$, each of length two. Recall that $v_{i}$ is the
only neighbor of $u_{i}$ in $G^{\prime }-N_{G^{\prime }}[v_{0}]$.
Furthermore recall that $z_{2j-1}$ is the only neighbor of $z_{2j}$ and $%
z_{2j}$ is the only neighbor of $z_{2j-1}$ in $N_{G^{\prime }}(v_{0})$.
Therefore $a=a^{\prime }=v_{i}$. That is, $z_{2j-1},z_{2j}\in N_{G^{\prime
}}(v_{i})$, which is a contradiction as we proved above. Therefore $%
|N_{G^{\prime }}(u_{i})\cap \{z_{2j-1},z_{2j}\}|\geq 1$. Similarly we can
prove that also $|N_{G^{\prime }}(v_{i})\cap \{z_{2j-1},z_{2j}\}|\geq 1$.
Summarizing, $|N_{G^{\prime }}(u_{i})\cap \{z_{2j-1},z_{2j}\}|=|N_{G^{\prime
}}(v_{i})\cap \{z_{2j-1},z_{2j}\}|=1$ for every $i\in \{1,2,\ldots
,k^{\prime }\}$ and every $j\in \{1,2,\ldots ,q\}$.

We construct a 2SAT formula $\phi $, i.e.~a boolean formula $\phi $ in
conjunctive normal form with two literals in every clause (2-CNF), such that 
$\phi $ is satisfiable if and only if $G^{\prime }$ is $3$-colorable. We
first define one boolean variable $x_{i}$ for every connected component $%
\{u_{i}^{\prime },v_{i}^{\prime }\}$ of $G^{\prime }-N_{G^{\prime }}[v_{0}]$%
, where~$1\leq i\leq k^{\prime }$ (cf.~line~\ref{alg-35}). If $x_{i}=1$,
then $u_{i}^{\prime }$ is colored red and $v_{i}^{\prime }$ is colored blue.
Otherwise, if~$x_{i}=0$, then $u_{i}^{\prime }$ is colored blue and $%
v_{i}^{\prime }$ is colored red. Furthermore, we define one boolean variable~%
$y_{j}$ for every pair $\{z_{2j-1},z_{2j}\}$ of $N_{1}$, where $1\leq j\leq
q $ (cf.~line~\ref{alg-35}). If $y_{j}=1$, then $z_{2j-1}$ is colored green
and $z_{2j}$ is colored blue. Otherwise, if $y_{j}=0$, then $z_{2j-1}$ is
colored blue and $z_{2j}$ is colored green.

Now, for every $j\in \{1,2,\ldots ,q\}$ and every vertex $u_{i}^{\prime }\in
V^{\prime }\setminus N_{G^{\prime }}[v_{0}]$, where $i\in \{1,2,\ldots
,k^{\prime }\}$, we add one clause to formula $\phi $ as follows. If $%
u_{i}^{\prime }z_{2j-1}\in E^{\prime }$, we add the clause $(x_{i}\vee 
\overline{y_{j}})$. Otherwise, if~${u_{i}^{\prime }z_{2j}\in E^{\prime}}$,
we add the clause $(x_{i}\vee y_{j})$, cf.~line~\ref{alg-36} of Algorithm~%
\ref{polynomial-L3-connected-diam-2-alg}. Furthermore, for every~${j\in
\{1,2,\ldots ,q\}}$ and every vertex $v_{i}^{\prime }\in V^{\prime
}\setminus N_{G^{\prime }}[v_{0}]$, where $i\in \{1,2,\ldots ,k^{\prime }\}$%
, we add one more clause to formula $\phi $ as follows. If $v_{i}^{\prime
}z_{2j-1}\in E^{\prime }$, we add the clause $(\overline{x_{i}}\vee 
\overline{y_{j}})$. Otherwise, if $v_{i}^{\prime }z_{2j}\in E^{\prime }$, we
add the clause $(\overline{x_{i}}\vee y_{j})$, cf.~line~\ref{alg-37} of
Algorithm~\ref{polynomial-L3-connected-diam-2-alg}.

By the above correspondence between truth values of the variables $%
x_{i},y_{j}$ and the colors of the vertices $u_{i}^{\prime },v_{i}^{\prime
},z_{2j-1},x_{2j}$, it is now easy to check that an arbitrary clause of the
formula $\phi $ is not satisfiable if and only if the corresponding edge
between one vertex of $\{u_{i}^{\prime },v_{i}^{\prime }\}$ and one vertex
of~$\{z_{2j-1},z_{2j}\}$ is monochromatic, i.e.~both its endpoints are
colored blue. Therefore, $\phi $ is satisfiable if and only if $G^{\prime }$
is $3$-colorable. Thus, if $\phi $ is not satisfiable, Algorithm~\ref%
{polynomial-L3-connected-diam-2-alg} correctly returns in line~\ref{alg-39}
that $G$ is not $3$-colorable. Otherwise, if $\phi $ is satisfiable,
Algorithm~\ref{polynomial-L3-connected-diam-2-alg} correctly computes a~$3$%
-coloring of the vertices of $G^{\prime }$ from a satisfying truth
assignment of $\phi $ (cf.~lines~\ref{alg-41a}-\ref{alg-45} of the
algorithm).

Summarizing, if the input graph $G$ is not $3$-colorable, Algorithm~\ref%
{polynomial-L3-connected-diam-2-alg} correctly returns a negative answer
during its execution (cf.~lines~\ref{alg-9},~\ref{alg-17}, and~\ref{alg-39}%
). Otherwise, Algorithm~\ref{polynomial-L3-connected-diam-2-alg} computes a
proper $3$-coloring $\chi ^{\prime }$ of the irreducible graph~$G^{\prime}$.
Given such a coloring of $G^{\prime }$, it is now straightforward to extend
it to a proper $3$-coloring of the input graph $G$ (cf.~line~\ref{alg-50} of
the algorithm). Namely, we have to iteratively undo all applications of the
Reduction Rules~\ref{rule1} and~\ref{rule2} that have previously transformed~%
$G$ into $G^{\prime }$ during the execution of lines~\ref{alg-13}-\ref%
{alg-14}, respectively. For every application of the Reduction Rule~\ref%
{rule1}, we color each of the original vertices with the color of the merged
vertex in $\chi ^{\prime }$. For every application of the Reduction Rule~\ref%
{rule2}, we color the removed vertex with the color of its sibling in $%
\chi^{\prime}$. Therefore, Algorithm~\ref{polynomial-L3-connected-diam-2-alg}
correctly returns in line~\ref{alg-51} the proper $3$-coloring of $G$ that
it has computed in line~\ref{alg-50}.

Regarding the time complexity of the algorithm, first recall that the size
of the formula $\phi $ is linear to the size of $G^{\prime }$ (and thus also
of $G$) and that the 2SAT problem can be solved in polynomial time.
Furthermore, recall that the Reduction Rules~\ref{rule1} and~\ref{rule2} can
be applied (as well as they can be undone) in polynomial time; the same
holds for the execution of all other lines of Algorithm~\ref%
{polynomial-L3-connected-diam-2-alg}. Therefore, Algorithm~\ref%
{polynomial-L3-connected-diam-2-alg} runs in polynomial time. This completes
the proof of the theorem.\qed
\end{proof}

A question that arises now naturally by Theorem~\ref%
{polynomial-L3-connected-diam-2-thm} is whether there exist any irreducible $%
3$-colorable graph~${G=(V,E)}$ with ${diam(G)=2}$, which has no articulation
neighborhood at all. A negative answer to this question would imply that we
can decide the $3$-coloring problem on \emph{any} graph with diameter~$2$ in
polynomial time using Algorithm~\ref{polynomial-L3-connected-diam-2-alg}.
However, the answer to that question is positive: for every $n\geq 1$, the
graph $G_{n}=(V_{n},E_{n})$ that has been presented in Theorem~\ref%
{example-irreducible-sqrt-n-thm} is irreducible, $3$-colorable, has diameter~%
$2$, and $G_{n}-N[v]$ is connected for every $v\in V_{n}$, i.e.~$G_{n}$ has
no articulation neighborhood. Therefore, Algorithm~\ref%
{polynomial-L3-connected-diam-2-alg} cannot be used in a trivial way to
decide in polynomial time the $3$-coloring problem for an arbitrary graph of
diameter~$2$. We leave the tractability of the $3$-coloring problem of
arbitrary diameter-$2$ graphs as an open~problem.

\section{Almost tight results for graphs with diameter~$3$\label%
{diameter-three-sec}}

In this section we present our results on graphs with diameter~$3$. In
particular, we first provide in Section~\ref{n-over-delta-diam-3-subsec} our
algorithm for $3$-coloring on arbitrary graphs with diameter~$3$ that has
running time~${2^{O(\min \{\delta \Delta ,\ \frac{n}{\delta }\log \delta \})}%
}$, where $n$ is the number of vertices and $\delta $ (resp.~$\Delta $) is
the minimum (resp.~maximum) degree of the input graph. Then we prove in
Section~\ref{diam-3-NP-c-subsec} that $3$-coloring is NP-complete on
irreducible and triangle-free graphs with diameter~$3$ and radius~$2$ by
providing a reduction from $3$SAT. Finally, we provide in Section~\ref%
{diam-3-amplification-subsec} our three different amplification techniques
that extend our hardness results of Section~\ref{diam-3-NP-c-subsec}. In
particular, we provide in Theorems~\ref%
{NP-complete-3-col-diam-3-large-min-degree-thm} and~\ref%
{NP-complete-3-col-diam-3-small-min-degree-thm} our NP-completeness
amplifications, and in Theorems~\ref{ETH-large-min-degree-thm},~\ref%
{ETH-medium-min-degree-thm}, and~\ref{ETH-small-min-degree-thm} our lower
bounds for the time complexity of $3$-coloring, assuming ETH.

\subsection{An $2^{O(\min \{\protect\delta \Delta ,\ \frac{n}{\protect\delta 
}\log \protect\delta \})}$-time algorithm for any graph with diameter~$3$%
\label{n-over-delta-diam-3-subsec}}

In the next theorem we use the DS-approach of Lemma~\ref%
{3-coloring-dominating-set-lem} to provide an improved $3$-coloring
algorithm for the case of graphs with diameter~$3$. The time complexity of
this algorithm is parameterized on the minimum degree $\delta $ of the given
graph $G$, as well as on the fraction $\frac{n}{\delta}$.

\begin{theorem}
\label{min-delta-Delta-n-over-delta-alg-diam-3-thm}Let ${G=(V,E)}$ be an
irreducible graph with $n$ vertices and ${diam(G)=3}$. Let~$\delta $ and~${%
\Delta }$ be the minimum and the maximum degree of $G$, respectively. Then,
the $3$-coloring problem can be decided in $2^{O(\min \{\delta \Delta ,\ 
\frac{n}{\delta }\log \delta \})}$ time on $G$.
\end{theorem}

\begin{proof}
First recall that, in an arbitrary graph $G$ with $n$ vertices and minimum
degree $\delta $, we can construct in polynomial time a dominating set $D$
with cardinality $|D|\leq n\frac{1+\ln (\delta +1)}{\delta +1}$~\cite%
{AlonProbabilistic08}. Therefore, similarly to the proof of Theorem~\ref%
{min-delta-n-over-delta-alg-diam-2-thm}, we can use the DS-approach of Lemma~%
\ref{3-coloring-dominating-set-lem} to obtain an algorithm that decides $3$%
-coloring on $G$ in $2^{O(\frac{n}{\delta }\log \delta )}$ time.

The DS-approach of Lemma~\ref{3-coloring-dominating-set-lem} applies to any
graph $G$. However, since $G$ has diameter~$3$ by assumption, we can design
a second algorithm for $3$-coloring of $G$ as follows. Consider a vertex ${%
u\in V}$ with minimum degree, i.e.~${\deg (u)=\delta}$. Since ${diam(G)=3}$,
we can partition the vertices of ${V\setminus N[u]}$ into two sets $A$ and $%
B $, such that ${A=\{v\in V\setminus N[u] : d(u,v)=2\}}$ and~${B=\{v\in
V\setminus N[u] : d(u,v)=3\}}$. Note that every vertex ${v\in A}$ is
adjacent to at least one vertex of $N(u)$. Therefore, since $|N(u)|=\delta $
and the maximum degree in $G$ is $\Delta $, it follows that~${|A|\leq \delta
\cdot \Delta}$. Furthermore, the set $A\cup \{u\}$ is a dominating set of $G$
with cardinality at most~${\delta \Delta +1}$. Thus we can decide $3$%
-coloring on $G $ by considering in worst case all possible $3$-colorings of~%
${A\cup \{u\}} $ in~$O^{\ast }(3^{\delta \Delta +1})=2^{O(\delta \Delta )}$
time by using the DS-approach of Lemma~\ref{3-coloring-dominating-set-lem}.

Summarizing, we can combine these two $3$-coloring algorithms for $G$,
obtaining an algorithm with time complexity $2^{O(\min \{\delta \Delta ,\ 
\frac{n}{\delta }\log \delta \})}$.\qed
\end{proof}

\medskip

To the best of our knowledge, the algorithm of Theorem~\ref%
{min-delta-Delta-n-over-delta-alg-diam-3-thm} is the first subexponential
algorithm for graphs with diameter~$3$, whenever $\delta =\omega (1)$, as
well as whenever $\delta =O(1)$ and $\Delta =o(n)$. As we will later prove
in Section~\ref{diam-3-amplification-subsec}, the running time provided in
Theorem~\ref{min-delta-Delta-n-over-delta-alg-diam-3-thm} is asymptotically
almost tight whenever $\delta =\Theta (n^{\varepsilon })$, for any $%
\varepsilon \in \lbrack \frac{1}{2},1)$.

Note now that for any graph $G$ with $n$ vertices and $diam(G)=3$, the
maximum degree $\Delta $ of~$G$ is~${\Delta =\Omega (n^{\frac{1}{3}})}$.
Indeed, suppose otherwise that $\Delta =o(n^{\frac{1}{3}})$, and let $u$ be
any vertex of~$G$. Then, there are at most $\Delta $ vertices in $G$ at
distance $1$ from $u$, at most $\Delta ^{2}$ vertices at distance $2$ from $%
u $, and at most~$\Delta ^{3}$ vertices at distance $3$ from $u$. That is, $%
G $ contains at most~${1+\Delta +\Delta ^{2}+\Delta ^{3}=o(n)}$ vertices,
since we assumed that ${\Delta =o(n^{\frac{1}{3}})}$, which is a
contradiction. Therefore~${\Delta =\Omega (n^{\frac{1}{3}})}$. Furthermore
note that, whenever ${\delta =\Omega (n^{\frac{1}{3}}\sqrt{\log n})}$ in a
graph with $n$ vertices and diameter~$3$, we have~${\delta \Delta =\Omega (%
\frac{n}{\delta }\log \delta )}$. Indeed, since ${\Delta =\Omega (n^{\frac{1%
}{3}})}$ as we proved above, it follows that in this case ${\delta
^{2}\Delta =\Omega (n\log n)=\Omega (n\log \delta )}$, since ${\delta <n}$,
and thus ${\delta \Delta =\Omega (\frac{n}{\delta }\log \delta )}$.
Therefore, if the minimum degree of $G$ is ${\delta =\Omega (n^{\frac{1}{3}}%
\sqrt{\log n})}$, the running time of the algorithm of Theorem~\ref%
{min-delta-Delta-n-over-delta-alg-diam-3-thm} becomes~${2^{O(\frac{n}{\delta 
}\log \delta )}=2^{O(n^{\frac{2}{3}}\sqrt{\log n})}}$.

\subsection{The $3$-coloring problem is NP-complete on graphs with diameter~$%
3$ and radius~$2$\label{diam-3-NP-c-subsec}}

In this section we provide a reduction from the $3$SAT problem to the $3$%
-coloring problem of triangle-free graphs with diameter~$3$ and radius~$2$.
Given a boolean formula $\phi $ in conjunctive normal form with three
literals in each clause (3-CNF), $\phi $ is \emph{satisfiable} if there is a
truth assignment of $\phi $, such that every clause contains at least one
true literal. The problem of deciding whether a given 3-CNF formula is
satisfiable, i.e.~the \emph{3SAT} problem, is one of the most known
NP-complete problems~\cite{GareyJohnson79}. Let $\phi $ be a 3-CNF formula
with $n$ variables $x_{1},x_{2},\ldots ,x_{n}$ and $m$ clauses $\alpha
_{1},\alpha _{2},\ldots ,\alpha _{m}$. We can assume in the following
without loss of generality that each clause has three distinct literals. We
now construct an irreducible and triangle-free graph $H_{\phi }=(V_{\phi
},E_{\phi })$ with diameter~$3$ and radius~$2$, such that $\phi $ is
satisfiable if and only if $H_{\phi }$ is $3$-colorable. Before we construct 
$H_{\phi }$, we first construct an auxiliary graph $G_{n,m}$ that depends
only on the number $n$ of the variables and the number $m$ of the clauses in 
$\phi $, rather than on $\phi $ itself.

We construct the graph ${G_{n,m}=(V_{n,m},E_{n,m})}$ as follows. Let $v_{0}$
be a vertex with~$8m$ neighbors $v_{1}, v_{2}, \ldots , v_{8m}$, which
induce an independent set. Consider also the sets ${U=\{u_{i,j}:1\leq i\leq
n+5m,1\leq j\leq 8m\}}$ and ${W=\{w_{i,j}:1\leq i\leq n+5m,1\leq j\leq 8m\}}$
of vertices. Each of these sets has $({n+5m)}8m$ vertices. The set $V_{n,m}$
of vertices of $G_{n,m}$ is defined as ${V_{n,m}=U\cup W\cup
\{v_{0},v_{1},v_{2},\ldots ,v_{8m}\}}$. That is,~${|V_{n,m}|=2\cdot ({n+5m}%
)8m+8m+1}$, and thus~$|V_{n,m}|=\Theta (m^{2})$, since $m=\Omega (n)$.

The set $E_{n,m}$ of the edges of $G_{n,m}$ is defined as follows. Define
first for every $j\in \{1,2,\ldots ,8m\}$ the subsets $U_{j}=%
\{u_{1,j},u_{2,j},\ldots ,u_{{n+5m},j}\}$ and $W_{j}=\{w_{1,j},w_{2,j},%
\ldots ,w_{{n+5m},j}\}$ of $U$ and $W$, respectively. Then define $%
N(v_{j})=\{v_{0}\}\cup U_{j}\cup W_{j}$ for every $j\in \{1,2,\ldots ,8m\}$,
where $N(v_{j})$ denotes the set of neighbors of vertex $v_{j}$ in $G_{n,m}$%
. For simplicity of the presentation, we arrange the vertices of $U\cup W$
on a rectangle matrix of size $2({n+5m)}\times 8m$, cf.~Figure~\ref%
{G-n-m-reduction-fig}(a). In this matrix arrangement, the $(i,j)$th element
is vertex $u_{i,j}$ if $i\leq {n+5m}$, and vertex $w_{i-{n-5m},j}$ if $i\geq 
{n+5m}+1$. In particular, for every $j\in \{1,2,\ldots ,8m\}$, the $j$th
column of this matrix contains exactly the vertices of $U_{j}\cup W_{j}$,
cf.~Figure~\ref{G-n-m-reduction-fig}(a). Note that, for every $j\in
\{1,2,\ldots ,8m\}$, vertex $v_{j}$ is adjacent to all vertices of the $j$th
column of this matrix. Denote now by $\ell _{i}=\{u_{i,1},u_{i,2},\ldots
,u_{i,8m}\}$ (resp.~$\ell _{i}^{\prime }=\{w_{i,1},w_{i,2},\ldots
,w_{i,8m}\} $) the $i$th (resp.~the $({n+5m}+i)$th) row of this matrix, cf.
Figure~\ref{G-n-m-reduction-fig}(a). For every $i\in \{1,2,\ldots ,{n+5m}\}$%
, the vertices of $\ell _{i}$ and of $\ell _{i}^{\prime }$ induce two
independent sets in $G_{n,m}$. We then add between the vertices of $\ell
_{i} $ and the vertices of $\ell _{i}^{\prime }$ all possible edges, except
those of $\{u_{i,j}w_{i,j}:1\leq j\leq 8m\}$. That is, we add all possible $%
(8m)^{2}-8m $ edges between the vertices of $\ell _{i}$ and of $\ell
_{i}^{\prime }$, such that they induce a complete bipartite graph without a
perfect matching between the vertices of $\ell _{i}$ and of $\ell
_{i}^{\prime }$, cf.~Figure~\ref{G-n-m-reduction-fig}(b). Note by the
construction of $G_{n,m}$ that both~$U$ and~$W$ are independent sets in $%
G_{n,m}$. Furthermore note that the minimum degree in~$G_{n,m}$ is $\delta
(G_{n,m})=\Theta (m)$ and the maximum degree is $\Delta (G_{n,m})=\Theta
(n+m)$. Thus, since $m=\Omega (n)$, we have that $\delta (G_{n,m})=\Delta
(G_{n,m})=\Theta (m)$. The construction of the graph $G_{n,m}$ is
illustrated in Figure~\ref{G-n-m-reduction-fig}.

\begin{figure}[h!tb]
\centering 
\includegraphics[scale=0.59]{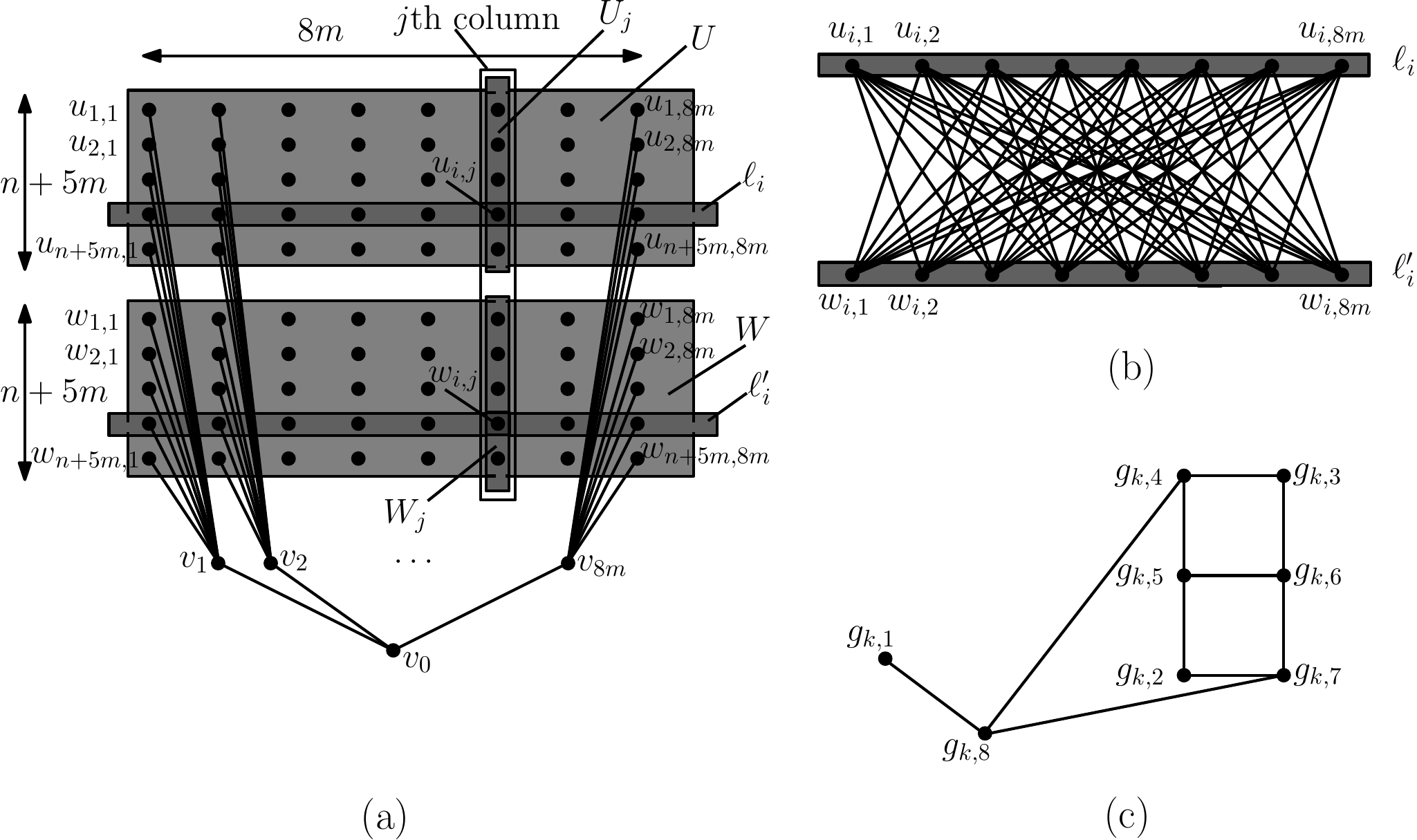}
\caption{(a)~The $2(n+5m)\times 8m$-matrix arrangement of the vertices $%
U\cup W$ of $G_{n,m}$ and their connections with the vertices $%
\{v_{0},v_{1},v_{2},\ldots ,v_{8m}\}$,~(b)~the edges between the vertices of
the $i$th row $\ell _{i}$ and the $(n+5m+i)$th row $\ell _{i}^{\prime }$ in
this matrix, and~(c)~the gadget with $8$ vertices and $10$ edges that we
associate in $H_{\protect\phi }$ to the clause $\protect\alpha _{k}$ of $%
\protect\phi $, where~${1\leq k\leq m}$.\protect\vspace{-0.4cm}}
\label{G-n-m-reduction-fig}
\end{figure}

\begin{lemma}
\label{G-phi-diameter-3-lem}For every $n,m\geq 1$, the graph $G_{n,m}$ has
diameter~$3$ and radius~$2$.
\end{lemma}

\begin{proof}
First note that for any $j\in \{1,2,\ldots ,8m\}$, every vertex $u_{i,j}$
(resp.~$w_{i,j}$) is adjacent to $v_{j}$. Therefore, since $v_{j}\in
N(v_{0}) $ for every $j\in \{1,2,\ldots ,8m\}$, it follows that $%
d(v_{0},u)\leq 2$ for every $u\in V_{n,m}\setminus \{v_{0}\}$, and thus $%
G_{n,m}$ has radius~$2 $. Furthermore note that $d(v_{j},v_{k})\leq 2$ for
every $j,k\in \{1,2,\ldots ,8m\}$, since $v_{j},v_{k}\in N(v_{0})$. Consider
now an arbitrary vertex~$u_{i,j}$ (resp.~$w_{i,j}$) and an arbitrary vertex $%
v_{k}$. If $j=k$, then ${d(v_{k},u_{i,j})=1}$ (resp.~${d(v_{k},w_{i,j})=1}$%
). Otherwise, if $j\neq k$, then there exists the path $%
(v_{k},v_{0},v_{j},u_{i,j})$ (resp.~$(v_{k},v_{0},v_{j},w_{i,j})$) of length 
$3$ between~$v_{k}$ and $u_{i,j}$ (resp.~$w_{i,j}$). Therefore also $%
d(v_{k},u)\leq 3$ for every $k=1,2,\ldots ,8m$ and every~${u\in
V_{n,m}\setminus \{v_{k}\}}$.

We will now prove that $d(u_{i,j},u_{p,q})\leq 3$ for every $(i,j)\neq (p,q)$%
. If $p\neq i$ and $q\neq j$, then there exists the path $%
(u_{i,j},v_{j},w_{p,j},u_{p,q})$ of length $3$ between $u_{i,j}$ and $%
u_{p,q} $. If $q=j$, then $d(u_{i,j},u_{p,q})=2$, as $u_{i,j},u_{p,q}\in
N(v_{j})$. Finally, if $p=i$, then there exists the path $%
(u_{i,j},v_{j},w_{i,j},u_{p,q})$ of length~$3$ between $u_{i,j}$ and $%
u_{p,q} $. Therefore $d(u_{i,j},u_{p,q})\leq 3$ for every $(i,j)\neq (p,q)$.
Similarly we can prove that~${d(w_{i,j},w_{p,q})\leq 3}$ for every $%
(i,j)\neq (p,q)$. It remains to prove that $d(u_{i,j},w_{p,q})\leq 3$ for
every $(i,j)$ and every $(p,q)$. If $p\neq i$ and $q\neq j$, then there
exists the path $(u_{i,j},v_{j},u_{p,j},w_{p,q})$ of length~$3$ between~$%
u_{i,j}$ and $w_{p,q}$. If $q=j$, then $d(u_{i,j},w_{p,q})=2$, since $%
u_{i,j},w_{p,q}\in N(v_{j})$. If $p=i$ and $q\neq j$, then $w_{p,q}=w_{i,q}$%
, and thus~$u_{i,j}w_{p,q}\in E_{\phi }$. Therefore $d(u_{ij},w_{pq})\leq 3$
for every $(i,j)\neq (p,q)$, and thus $G_{n,m}$ has diameter~$3$.\qed
\end{proof}

\begin{lemma}
\label{G-phi-irreducible-lem}For every $n,m\geq 1$, the graph $G_{n,m}$ is
irreducible and triangle-free.
\end{lemma}

\begin{proof}
First observe that, by construction, $G_{n,m}$ has no pair of sibling
vertices, and thus the Reduction Rule~\ref{rule2} does not apply to $G_{n,m}$%
. Thus, in order to prove that $G_{n,m}$ is irreducible, it suffices to
prove that $G_{n,m}$ is triangle-free, and thus also diamond-free, i.e.~also
the Reduction Rule~\ref{rule1} does not apply to $G_{n,m}$. Suppose
otherwise that $G_{n,m}$ has a triangle with vertices $a,b,c$. Note that
vertex $v_{0}$ does not belong to any triangle in $G_{n,m}$, since its
neighbors $N(v_{0})=\{v_{1},v_{2},\ldots ,v_{8m}\}$ induce an independent
set in $G_{n,m}$. Furthermore, note that also vertex $v_{j}$, where $1\leq
j\leq 8m$, does not belong to any triangle in $G_{n,m}$. Indeed, by
construction of $G_{n,m}$, its neighbors $N(v_{j})=\{v_{0}\}\cup U_{j}\cup
W_{j}$ induce an independent set in $G_{n,m}$. Therefore all the vertices $%
a,b,c$ of the assumed triangle of $G_{n,m}$ belong to $U\cup W$. Therefore,
at least two vertices among $a,b,c$ belong to $U$, or at least two of them
belong to $W$. This is a contradiction, since by the construction of $%
G_{n,m} $ the sets $U$ and $W$ induce two independent sets. Therefore $%
G_{n,m}$ is triangle-free. Thus $G_{n,m}$ is also diamond-free, i.e.~the
Reduction Rule \ref{rule1} does not apply to $G_{n,m}$. Summarizing, $%
G_{n,m} $ is irreducible and triangle-free.\qed
\end{proof}

We now construct the graph $H_{\phi }=(V_{\phi },E_{\phi })$ from $\phi $ by
adding $10m$ edges to $G_{n,m}$ as follows. Let $k\in \{1,2,\ldots ,m\}$ and
consider the clause $\alpha _{k}=(l_{k,1}\vee l_{k,2}\vee l_{k,3})$, where $%
l_{k,p}\in \{x_{i_{k,p}},\overline{x_{i_{k,p}}}\}$ for $p\in \{1,2,3\}$ and $%
i_{k,1},i_{k,2},i_{k,3}\in \{1,2,\ldots ,n\}$. For this clause $\alpha _{k}$%
, we add on the vertices of $G_{n,m}$ an isomorphic copy of the gadget in
Figure~\ref{G-n-m-reduction-fig}(c), which has $8$ vertices and $10$ edges,
as follows. Let $p\in \{1,2,3\}$. The literal $l_{k,p}$ corresponds to
vertex $g_{k,p}$ of this gadget. If $l_{k,p}=x_{i_{k,p}}$, we set $%
g_{k,p}=u_{i_{k,p},8k+1-p}$. Otherwise, if $l_{k,p}=\overline{x_{i_{k,p}}}$,
we set $g_{p}=w_{i_{k,p},8k+1-p}$. Furthermore, for $p\in \{4,\ldots ,8\}$,
we set $g_{k,p}=u_{n+5k+4-p,8k+1-p}$.

Note that, by construction, the graphs $H_{\phi }$ and $G_{n,m}$ have the
same vertex set, i.e.~$V_{\phi }=V_{n,m}$, and that $E_{n,m}\subset E_{\phi
} $. Therefore $diam(H_{\phi })=3$ and $rad(H_{\phi })=2$, since $%
diam(G_{n,m})=3$ and $rad(G_{n,m})=2$ by Lemma~\ref{G-phi-diameter-3-lem}.
Observe now that every positive literal of $\phi $ is associated to a vertex
of $U$, while every negative literal of $\phi $ is associated to a vertex of 
$W$. In particular, each of the $3m$ literals of $\phi $ corresponds by this
construction to a different column in the matrix arrangement of the vertices
of $U\cup W$. If a literal of $\phi $ is the variable $x_{i}$ (resp.~the
negated variable $\overline{x_{i}}$), where $1\leq i\leq n$, then the vertex
of $U$ (resp.~$W$) that is associated to this literal lies in the $i$th row $%
\ell _{i}$ (resp.~in the $(n+5m+i)$th row $\ell _{i}^{\prime }$) of the
matrix. Moreover, note by the above construction that each of the $8m$
vertices $\{q_{k,1},q_{k,2},\ldots ,q_{k,8}\}_{k=1}^{m}$ corresponds to a
different column in the matrix of the vertices of $U\cup W$. Finally, each
of the $5m$ vertices $\{q_{k,4},q_{k,5},q_{k,6},q_{k,7},q_{k,8}\}_{k=1}^{m}$
corresponds to a different row in the matrix of the vertices of $U$.

\begin{observation}
\label{gadget-no-2-coloring-obs}The gadget of Figure~\ref%
{G-n-m-reduction-fig}(c) has no proper $2$-coloring, as it contains an
induced cycle of length $5$.
\end{observation}

\begin{observation}
\label{gadget-coloring-obs}Consider the gadget of Figure~\ref%
{G-n-m-reduction-fig}(c). If we assign to vertices $g_{k,1},g_{k,2},g_{k,3}$
the same color, we cannot extend this coloring to a proper $3$-coloring of
the gadget. Furthermore, if we assign to vertices $g_{k,1},g_{k,2},g_{k,3}$
in total two or three colors, then we can extend this coloring to a proper $%
3 $-coloring of the gadget.
\end{observation}

The next observation follows by the construction of $H_{\phi }$ and by our
initial assumption that each clause of $\phi $ has three distinct literals.

\begin{observation}
\label{no-adj-pair-in-line-obs}For every $i\in \{1,2,\ldots ,n+5m\}$, there
exists no pair of adjacent vertices in the same row $\ell _{i}$ or $\ell
_{i}^{\prime }$ in $H_{\phi }$.
\end{observation}

\begin{lemma}
\label{H-phi-irreducible-lem}For every formula $\phi $, the graph $H_{\phi }$
is irreducible and triangle-free.
\end{lemma}

\begin{proof}
First observe that, similarly to $G_{n,m}$, the graph $H_{\phi }$ has no
pair of sibling vertices, and thus the Reduction Rule~\ref{rule2} does not
apply to $H_{\phi }$. We will now prove that $H_{\phi }$ is triangle-free.
Suppose otherwise that $H_{\phi }$ has a triangle with vertices $a,b,c$.
Similarly to $G_{n,m}$, note that the neighbors of $v_{0}$ in $H_{\phi }$
induce an independent set, and thus vertex $v_{0}$ does not belong to any
triangle in $H_{\phi }$. Furthermore, note by the construction of $H_{\phi }$
that we do not add any edge between vertices $u_{i,j}$ and $w_{i,j}$, where $%
1\leq i\leq n+5m$ and $1\leq j\leq 8m$. Therefore, similarly to $G_{n,m}$,
the neighbors of $v_{j}$ induce an independent set in $H_{\phi }$, where $%
j\in \{1,2,\ldots ,8m\}$, and thus $v_{j}$ does not belong to any triangle
in $H_{\phi }$. Therefore all the vertices $a,b,c$ of the assumed triangle
of $H_{\phi }$ belong to $U\cup W$. Since $G_{n,m}$ is triangle-free by
Lemma~\ref{G-phi-irreducible-lem}, it follows that at least one edge in this
triangle belongs to $E_{\phi }\setminus E_{n,m}$, i.e.~to at least one of
the copies of the gadget in Figure~\ref{G-n-m-reduction-fig}(c). Note that
not all vertices $a,b,c$ belong to the same copy of this gadget, since the
gadget is triangle-free, cf.~Figure~\ref{G-n-m-reduction-fig}(c). Assume
without loss of generality that vertices $a$ and $b$ (and thus also the edge 
$ab$) of the assumed triangle of $H_{\phi }$ belong to the copy of the
gadget that corresponds to clause $\alpha _{k}$, where $1\leq k\leq m$.
Then, since vertex $c$ does not belong to this gadget, the edges $ac$ and $%
bc $ of the assumed triangle belong also to the graph $G_{n,m}$. Therefore,
since $a,b,c\in U\cup W$, it follows by the construction of $G_{n,m}$ that $%
a $ and $b$ belong to some row $\ell _{i}$ and $c$ belongs to the row $\ell
_{i}^{\prime }$, or $a$ and $b$ belong to some row $\ell _{i}^{\prime }$ and 
$c$ belongs to the row $\ell _{i}$. This is a contradiction, since no pair
of adjacent vertices (such as $a$ and $b$) belong to the same row $\ell _{i}$
or $\ell _{i}^{\prime }$ in $H_{\phi }$ by Observation \ref%
{no-adj-pair-in-line-obs}. Therefore $H_{\phi }$ is triangle-free, and thus
also diamond-free, i.e. the Reduction Rule~\ref{rule2} does not apply to $%
H_{\phi }$. Summarizing, $H_{\phi }$ is irreducible and triangle-free.\qed
\end{proof}

\medskip

We are now ready to state the main theorem of this section.

\begin{theorem}
\label{3-col-diam-3-SAT-thm}The formula $\phi $ is satisfiable if and only
if $H_{\phi }$ is $3$-colorable.
\end{theorem}

\begin{proof}
We will first prove that $G_{n,m}$ is always $3$-colorable. Recall that both 
$U$ and $W$ are independent sets in $G_{n,m}$, and that the only edges among
the vertices of $U\cup W$ in $G_{n,m}$ are all possible edges between the
rows $\ell _{i}$ (that contains only vertices of $U$) and $\ell _{i}^{\prime
}$ (that contains only vertices of $W$), except of a perfect matching
between the vertices of $\ell _{i}$ and of $\ell _{i}^{\prime }$. Consider
three colors, say red, green, and blue. We assign to vertex $v_{0}$ the
color red. Furthermore we assign arbitrarily the color blue or green to each
of its neighbors~$v_{j}$, $1\leq j\leq 8m$. For each of these $2^{8m}$
different colorings of vertex $v_{0}$ and its neighbors, we can construct~$%
2^{n+5m}$ different proper $3$-colorings of $G_{n,m}$ as follows. For every $%
i\in \{1,2,\ldots ,n+5m\}$, we have at least two possibilities of coloring
the vertices of $\ell _{i}$ and of $\ell _{i}^{\prime }$: (a) color all
vertices of $\ell _{i}$ red, and for every vertex $w_{i,j}$ of $\ell
_{i}^{\prime }$, color $w_{i,j}$ blue (resp.~green) if $v_{j}$ is colored
green (resp.~blue), and (b) color all vertices of $\ell _{i}^{\prime }$ red,
and for every vertex $u_{i,j}$ of $\ell _{i}$, color $u_{i,j}$ blue
(resp.~green) if $v_{j}$ is colored green (resp.~blue). Therefore, there are
at least $2^{8m}\cdot 2^{n+5m}$ different proper $3$-colorings of~$G_{n,m}$,
in which vertex $v_{0}$ obtains color red.

\medskip

($\Rightarrow $) Suppose first that $\phi $ is satisfiable, and let $\tau $
be a satisfying truth assignment of $\phi $. Given this truth assignment $%
\tau $, we construct a proper $3$-coloring $\chi _{\phi }$ of $H_{\phi }$ as
follows. First assign to $v_{0}$ the color red in $\chi _{\phi }$. By
construction, this coloring $\chi _{\phi }$ will be one of the above $%
2^{8m}\cdot 2^{n+5m}$ proper $3$-colorings of $G_{n,m}$. That is, for every $%
i\in \{1,2,\ldots ,n+5m\}$, either all vertices of row $\ell _{i}$ are red
and all vertices of row $\ell _{i}^{\prime }$ are green or blue in $\chi
_{\phi }$, or vice versa. For every $i\in \{1,2,\ldots ,n+5m\}$, if the
vertices of row $\ell _{i}$ (resp.~of row $\ell _{i}^{\prime }$) are red in $%
\chi _{\phi }$, then we call row $\ell _{i}$ (resp.~row $\ell _{i}^{\prime }$%
) a \emph{red line}. Otherwise, if the vertices of row $\ell _{i}$ (resp. of
row $\ell _{i}^{\prime }$) are not red in $\chi _{\phi }$, then we call row $%
\ell _{i}$ (resp.~row $\ell _{i}^{\prime }$) a \emph{white line}.
Furthermore, for every vertex $u_{i,j}$ (resp.~$w_{i,j}$) of a white line $%
\ell _{i}$ (resp.~of a white line $\ell _{i}^{\prime }$), the color of $%
u_{i,j}$ (resp.~of $w_{i,j}$) in $\chi _{\phi }$ is uniquely determined by
the color of its neighbor $v_{j}$ in $\chi _{\phi }$. That is, if $v_{j}$ is
blue then $u_{i,j}$ (resp.~$w_{i,j}$) is green in $\chi _{\phi }$.
Otherwise, if $v_{j}$ is green then $u_{i,j}$ (resp.~$w_{i,j}$) is blue in $%
\chi _{\phi }$.

Let $x_{i}$ be an arbitrary variable in $\phi $, where $1\leq i\leq n$. If $%
x_{i}=0$ in $\tau $, we define row $\ell _{i}$ to be a red line and row $%
\ell _{i}^{\prime }$ to be a white line in $\chi _{\phi }$, respectively.
Otherwise, if $x_{i}=1$ in $\tau $, we define row $\ell _{i}^{\prime }$ to
be a red line and row $\ell _{i}$ to be a white line in $\chi _{\phi }$,
respectively. Consider an arbitrary clause $\alpha _{k}=(l_{k,1}\vee
l_{k,2}\vee l_{k,3})$ of $\phi $, where $l_{k,p}\in \{x_{i_{k,p}},\overline{%
x_{i_{k,p}}}\}$ for $p\in \{1,2,3\}$. Furthermore consider the copy of the
gadget of Figure~\ref{G-n-m-reduction-fig}(c) that is associated to clause $%
\alpha _{k}$ in $H_{\phi }$. Recall by the construction of $H_{\phi }$ that
the literals $l_{k,1}$, $l_{k,2}$, and $l_{k,3}$ correspond to the vertices $%
q_{k,1}$, $q_{k,2}$, and $q_{k,3}$ of this gadget, respectively. Since $\tau 
$ is a satisfying assignment of $\phi $, at least one of the literals $%
l_{k,1}$, $l_{k,2}$, and $l_{k,3}$ is true in $\tau $. Therefore at least
one of the vertices $q_{k,1}$, $q_{k,2}$, and $q_{k,3}$ belongs to a white
line in $\chi _{\phi }$, i.e.~at least one of them is green or blue in $\chi
_{\phi }$. If one (resp.~two) of the vertices $q_{k,1},q_{k,2},q_{k,3}$
belongs (resp.~belong) to a red line of $\chi _{\phi }$, then we color the
other two (resp.~the other one) green in $\chi _{\phi }$. Otherwise, if all
three of the vertices $q_{k,1}$, $q_{k,2}$, and $q_{k,3}$ belong to a white
line in $\chi _{\phi }$, then we color $q_{k,1},q_{k,2}$ green and vertex $%
q_{k,3}$ blue in $\chi _{\phi }$. For every $p\in \{1,2,3\}$, if we color
vertex $q_{k,p}$ green (resp.~blue) in $\chi _{\phi }$, then we color its
neighbor $v_{8k+1-p}$ blue (resp.~green) in $\chi _{\phi }$, cf.~the
construction of the graph $H_{\phi }$. Otherwise, if we color vertex $%
q_{k,p} $ red in $\chi _{\phi }$, then we color its neighbor $v_{8k+1-p}$
either blue or green in $\chi _{\phi }$ (both choices lead to a proper $3$%
-coloring of $H_{\phi }$).

Once we have colored the vertices $q_{k,1},q_{k,2},q_{k,3}$ with two colors
in total, we extend the coloring of these three vertices to a proper $3$%
-coloring $\chi _{k}$ of the gadget associated to clause $\alpha _{k}$ (cf.
Observation~\ref{gadget-coloring-obs}). Let $p\in \{4,5,6,7,8\}$. If $%
g_{k,p} $ is colored green (resp.~blue) in $\chi _{k}$, then we color its
neighbor $v_{8k+1-p}$ blue (resp.~green) in $\chi _{\phi }$, and we define
row $\ell _{i}$ to be a white line and row $\ell _{i}^{\prime }$ to be a red
line in $\chi _{\phi }$, respectively. Otherwise, if $g_{k,p}$ is colored
red in $\chi _{k}$, then we define row $\ell _{i}$ to be a red line and row $%
\ell _{i}^{\prime }$ to be a white line in $\chi _{\phi }$, respectively.
Furthermore, in this case we color the neighbor $v_{8k+1-p}$ of $g_{k,p}$
either blue or green $\chi _{\phi }$ (both choices lead to a proper $3$%
-coloring of $H_{\phi }$). After performing the above coloring operations
for every clause $\alpha _{k}$, where $1\leq k\leq m$, we obtain a well
defined coloring $\chi _{\phi }$ of all vertices of $H_{\phi }$. Note that
in this coloring $\chi _{\phi }$, all copies of the gadget of Figure~\ref%
{G-n-m-reduction-fig}(c) are properly colored with at most $3$ colors.
Furthermore, it can be easily verified that this coloring $\chi _{\phi }$ is
also a proper $3$-coloring of $G_{n,m}$. Therefore $\chi _{\phi }$ is a
proper $3$-coloring~of~$H_{\phi }$.

\medskip

($\Leftarrow $) Suppose now that $H_{\phi }$ is $3$-colorable and let $\chi
_{\phi }$ be a proper $3$-coloring of $H_{\phi }$. Assume without loss of
generality that vertex $v_{0}$ is colored red in $\chi _{\phi }$. We will
construct a satisfying assignment $\tau $ of $\phi $. Consider an index $%
i\in \{1,2,\ldots ,n+5m\}$ and the rows $\ell _{i}$ and $\ell _{i}^{\prime }$
of the matrix. Suppose that $\ell _{i}$ has at least one vertex $u_{i,j_{1}}$
that is colored red and at least one vertex $u_{i,j_{2}}$ that is colored
blue in $\chi _{\phi }$. Then clearly all vertices of $\ell _{i}^{\prime }$,
except possibly of $w_{i,j_{1}}$ and $w_{i,j_{2}}$, are colored green in $%
\chi _{\phi }$, since they are adjacent to both $u_{i,j_{1}}$ and $%
u_{i,j_{2}}$. Therefore all vertices of $\{v_{1},v_{2},\ldots
,v_{8m}\}\setminus \{v_{j_{1}},v_{j_{2}}\}$ are colored blue in $\chi _{\phi
}$. Thus, all vertices of $(U\cup W)\setminus (U_{j_{1}}\cup U_{j_{2}}\cup
W_{j_{1}}\cup W_{j_{2}})$ are colored either green or red in $\chi _{\phi }$%
. However there exists at least one copy of the gadget of Figure~\ref%
{G-n-m-reduction-fig}(c) on these vertices, by the construction of $H_{\phi
} $. That is, this gadget has a proper coloring (induced by $\chi _{\phi }$)
with the colors green and red. This is a contradiction by Observation~\ref%
{gadget-no-2-coloring-obs}. Thus, there exists no row $\ell _{i}$ with at
least one vertex colored red and another one colored blue in $\chi _{\phi }$%
. Similarly we can prove that there exists no row $\ell _{i}$ (resp.~$\ell
_{i}^{\prime }$) with at least one vertex colored red and another one
colored blue or green in $\chi _{\phi }$. That is, if at least one vertex of
a row $\ell _{i}$ (resp.~$\ell _{i}^{\prime }$) is colored red in $\chi
_{\phi }$, then all vertices of $\ell _{i}$ (resp.~$\ell _{i}^{\prime }$)
are colored red in~$\chi _{\phi }$.

We will now prove that for any $i\in \{1,2,\ldots ,n+5m\}$, at least one
vertex of $\ell _{i}$ or at least one vertex of $\ell _{i}^{\prime }$ is red
in $\chi _{\phi }$. Suppose otherwise that every vertex of the rows $\ell
_{i}$ and $\ell _{i}^{\prime }$ is colored either green or blue in $\chi
_{\phi }$. Then, since the vertices of $\ell _{i}$ and of $\ell _{i}^{\prime
}$ induce a connected bipartite graph, it follows that all vertices of $\ell
_{i}$ are colored green and all vertices of $\ell _{i}^{\prime }$ are
colored blue in $\chi _{\phi }$, or vice versa. Thus, in particular, for
every $j=1,2,\ldots ,8m$, vertex $v_{j}$ is adjacent to one blue and to one
green vertex (one from $\ell _{i}$ and the other one from $\ell _{i}^{\prime
}$). Thus, since $\chi _{\phi }$ is a proper $3$-coloring of $H_{\phi }$, it
follows that $v_{j}$ is colored red in $\chi _{\phi }$. This is a
contradiction, since $v_{0}\in N(v_{j})$ and $v_{0}$ is colored red in $\chi
_{\phi }$ by assumption. Therefore, for any $i\in \{1,2,\ldots ,n+5m\}$, at
least one vertex of $\ell _{i}$ or at least one vertex of $\ell _{i}^{\prime
}$ is colored red in $\chi _{\phi }$.

Summarizing, for every ${i\in \{1,2,\ldots ,n+5m\}}$, either all vertices of
the row $\ell _{i}$ or all vertices of the row $\ell _{i}^{\prime }$ are
colored red in $\chi _{\phi }$. We define now the truth assignment $\tau $
of $\phi $ as follows. For every ${i\in \{1,2,\ldots ,n+5m\}}$, we set ${%
x_{i}=0}$ in $\tau $ if all vertices of $\ell _{i}$ are colored red in $\chi
_{\phi }$; otherwise, if all vertices of $\ell _{i}^{\prime }$ are colored
red in $\chi _{\phi }$, then we set ${x_{i}=1}$ in $\tau $. We will prove
that $\tau $ is a satisfying assignment of $\phi $. Consider a clause $%
\alpha _{k}=(l_{k,1}\vee l_{k,2}\vee l_{k,3})$ of $\phi $, where $l_{k,p}\in
\{x_{i_{k,p}},\overline{x_{i_{k,p}}}\}$ for $p\in \{1,2,3\}$. By the
construction of the graph $H_{\phi }$, this clause corresponds to a copy of
the gadget of Figure~\ref{G-n-m-reduction-fig}(c) in $H_{\phi }$. Thus,
since $\chi _{\phi }$ is a proper $3$-coloring of $H_{\phi }$ by assumption,
the vertices of this gadget are colored with three colors by Observation~\ref%
{gadget-no-2-coloring-obs}. Furthermore, not all three vertices $%
q_{k,1},q_{k,2},q_{k,3}$ have the same color in $\chi _{\phi }$ by
Observation~\ref{gadget-coloring-obs}. Moreover, note by the construction of 
$H_{\phi }$ and by the definition of the truth assignment $\tau $, that $%
l_{k,p}=0$ in $\tau $ if and only if vertex $q_{k,p}$ is colored red in $%
\chi _{\phi }$, where $p\in \{1,2,3\}$. Thus, since the vertices $%
q_{k,1},q_{k,2},q_{k,3}$ are not all red in $\chi _{\phi }$, it follows that
the literals $l_{k,1},l_{k,2},l_{k,3}$ of clause $\alpha _{k}$ are not all
false in $\tau $. Therefore $\alpha _{k}$ is satisfied by $\tau $, and thus $%
\tau $ is a satisfying truth assignment of $\phi $. This completes the proof
of the theorem.\qed
\end{proof}

\medskip

The next theorem follows by Lemma~\ref{H-phi-irreducible-lem} and Theorem~%
\ref{3-col-diam-3-SAT-thm}.

\begin{theorem}
\label{NP-complete-3-col-diam-3-thm}The $3$-coloring problem is NP-complete
on irreducible and triangle-free graphs with diameter~$3$ and radius~$2$.
\end{theorem}

\subsection{Lower time complexity bounds and general NP-completeness results 
\label{diam-3-amplification-subsec}}

In this section we present our three different amplification techniques of
the reduction of Theorem~\ref{3-col-diam-3-SAT-thm}. In particular, using
these three amplifications we extend for every~${\varepsilon \in \lbrack 0,1)%
}$ the result of Theorem~\ref{NP-complete-3-col-diam-3-thm} (by providing
both NP-completeness and lower time complexity bounds) to irreducible
triangle-free graphs with diameter~$3$ and radius~$2$ and minimum degree ${%
\delta =\Theta (n^{\varepsilon })}$. We use our first amplification
technique in Theorems~\ref{NP-complete-3-col-diam-3-large-min-degree-thm}
and~\ref{ETH-large-min-degree-thm}, our second one in Theorems~\ref%
{NP-complete-3-col-diam-3-small-min-degree-thm} and~\ref%
{ETH-small-min-degree-thm}, and our third one in~Theorem~\ref%
{ETH-medium-min-degree-thm}.

\begin{theorem}
\label{NP-complete-3-col-diam-3-large-min-degree-thm}Let ${G=(V,E)}$ be an
irreducible and triangle-free graph with diameter~$3$ and radius~$2$. If the
minimum degree of $G$ is ${\delta (G)=\Theta (|V|^{\varepsilon })}$, where ${%
\varepsilon \in \lbrack \frac{1}{2},1)}$, then it is NP-complete to decide
whether~$G$ is~$3$-colorable.
\end{theorem}

\begin{proof}
Let $\phi $ be a boolean formula with $n$ variables and $m$ clauses. Using
the reduction of Section~\ref{diam-3-NP-c-subsec}, we construct from the
formula $\phi $ the irreducible and triangle-free graph $H_{\phi }=(V_{\phi
},E_{\phi })$, such that $H_{\phi }$ has diameter~$3$ and radius~$2$.
Furthermore $|V_{\phi }|=\Theta (nm)$ by the construction of $H_{\phi }$.
Then, $\phi $ is satisfiable if and only if $H_{\phi }$ is $3$-colorable by
Theorem~\ref{3-col-diam-3-SAT-thm}.

Let now ${\varepsilon \in \lbrack \frac{1}{2},1)}$. Define $\varepsilon _{0}=%
\frac{\varepsilon }{1-\varepsilon }$ and $k_{0}=m^{\varepsilon _{0}}$. Since 
${\varepsilon \in \lbrack \frac{1}{2},1)}$ by assumption, it follows that $%
\varepsilon _{0}\geq 1$. We construct now from the graph $H_{\phi }$ the
irreducible graph $H_{1}(\phi ,\varepsilon )$ with diameter~$3$ and radius~$%
2 $, as follows. For every vertex $v_{j}$ in $H_{\phi }$, where $j\in
\{1,2,\ldots ,8m\}$, remove $v_{j}$ and introduce the $2k_{0}$ new vertices $%
A_{j}=\{v_{j,1}^{\prime },v_{j,2}^{\prime },\ldots ,v_{j,k_{0}}^{\prime }\}$
and $B_{j}=\{v_{j,1}^{\prime \prime },v_{j,2}^{\prime \prime },\ldots
,v_{j,k_{0}}^{\prime \prime }\}$. The vertices of $A_{j}$ and of $B_{j}$
induce two independent sets in $H_{1}(\phi ,\varepsilon )$. We then add
between the vertices of $A_{j}$ and of $B_{j}$ all possible edges, except
those of $\{v_{j,p}^{\prime }v_{j,p}^{\prime \prime }:1\leq p\leq k_{0}\}$.
That is,we add $k_{0}^{2}-k_{0}$ edges between the vertices of $A_{j}$ and $%
B_{j}$, such that they induce a complete bipartite graph without a perfect
matching between $A_{j}$ and $B_{j}$. This replacement of vertex $v_{j}$ by
the vertex sets~$A_{j}$ and~$B_{j}$ in~$H_{1}(\phi ,\varepsilon )$ is
illustrated in Figure~\ref{first-amplification-construction-fig}.
Furthermore, we add all $2k_{0}$ edges between $v_{0}$ and the vertices of $%
A_{j}\cup B_{j}$. Finally, we add all $2(n+5m)\cdot k_{0}$ edges between the 
$2(n+5m)$ vertices of $U_{j}\cup W_{j}$ (i.e.~the $j$th column of the matrix
arrangement of the vertices of $U\cup V$) and the $k_{0}$ vertices of $A_{j}$%
. Denote the resulting graph by $H_{1}(\phi ,\varepsilon )$.

\begin{figure}[h!tb]
\centering 
\subfigure[]{ \label{first-amplification-construction-fig-1}
\includegraphics[scale=0.41]{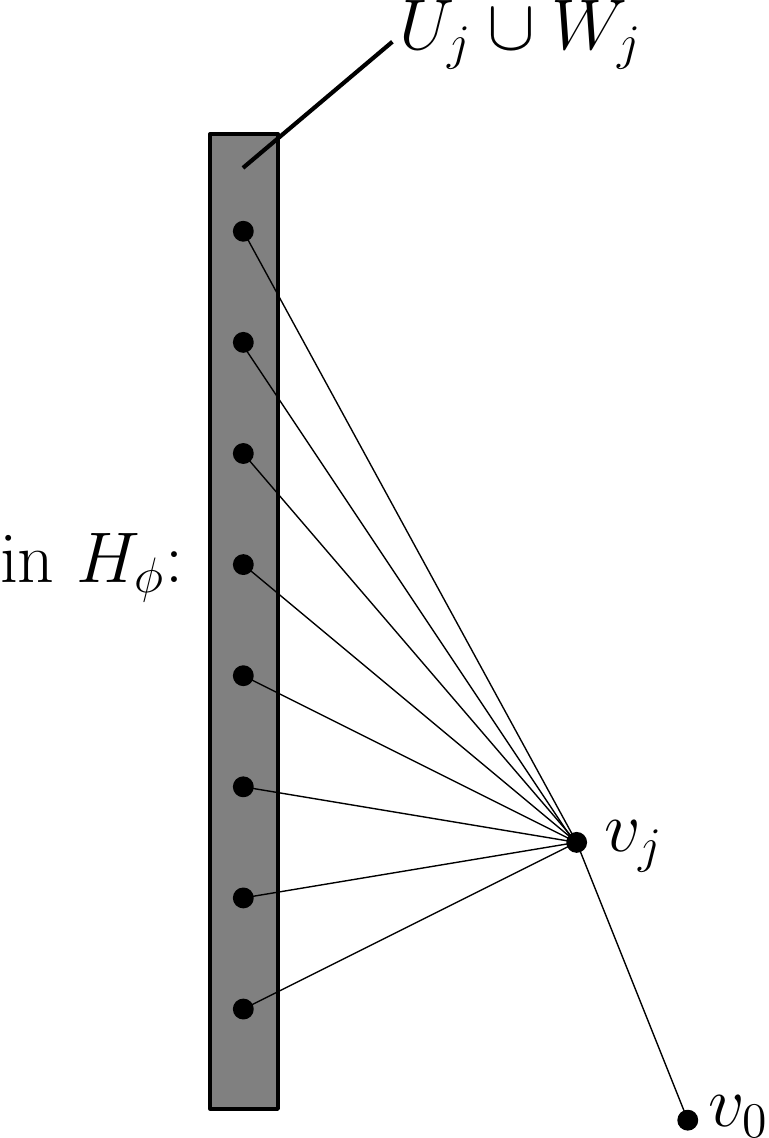}} 
\hspace{0.8cm} 
\subfigure[]{ \label{first-amplification-construction-fig-2}
\includegraphics[scale=0.41]{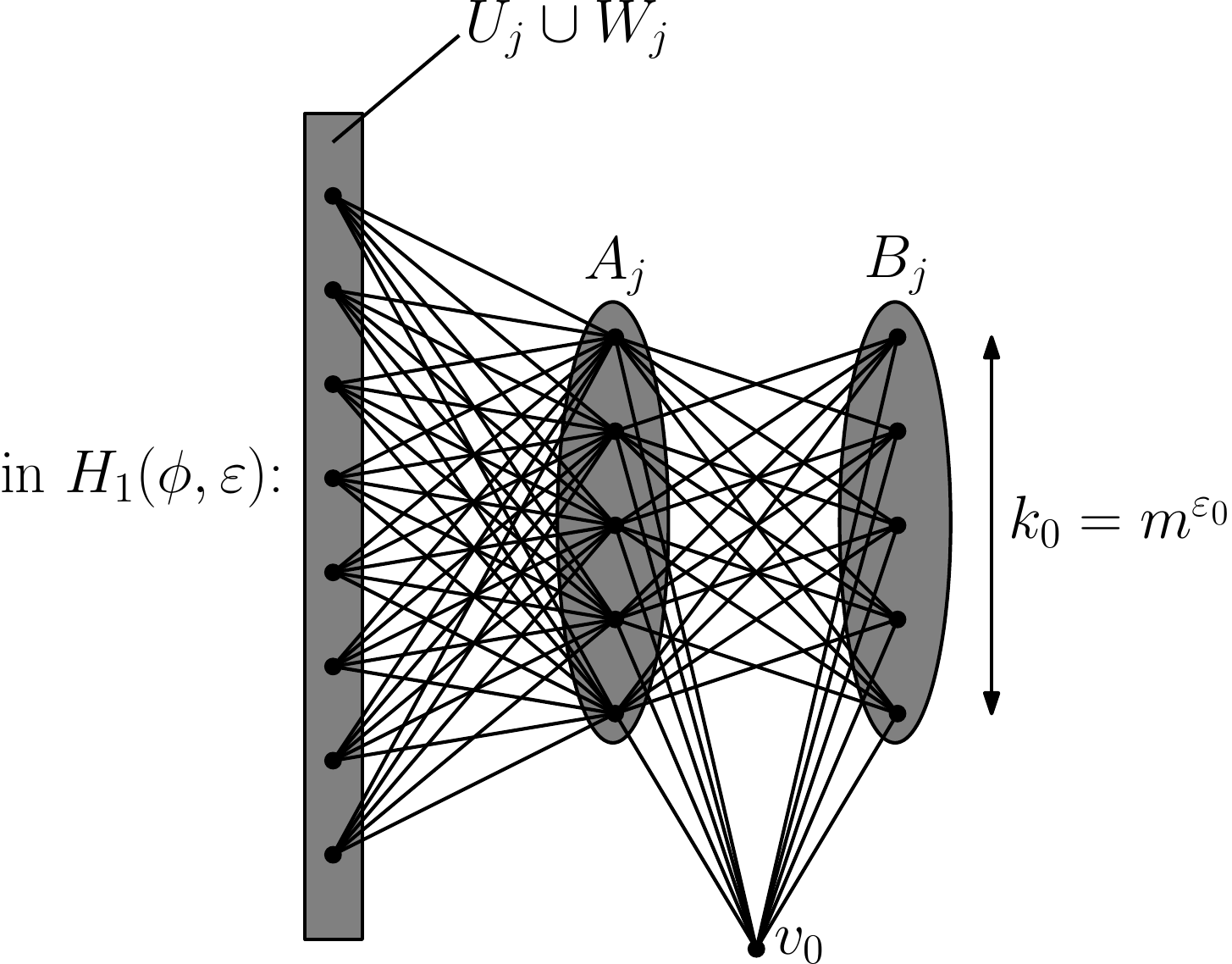}}
\caption{(a)~The vertex $v_{j}$ with its neighbors $N(v_{j})=\{v_{0}\} \cup
U_{j} \cup W_{j}$ in the graph $H_{\protect\phi}$ and~(b)~the vertex sets~$%
A_{j}$ and~$B_{j}$ that replace vertex $v_{j}$ in the graph $H_{1}(\protect%
\phi,\protect\varepsilon)$.}
\label{first-amplification-construction-fig}
\end{figure}

Observe that, by this construction, for every $j\in \{1,2,\ldots ,8m\}$, all
neighbors of vertex $v_{j}$ in the graph $H_{\phi }$ are included in the
neighborhood of every vertex $v_{j,p}^{\prime }$ of $A_{j}$ in the graph $%
H_{1}(\phi ,\varepsilon )$, where $1\leq p\leq k_{0}$. In particular, $%
H_{\phi }$ is an induced subgraph of $H_{1}(\phi ,\varepsilon )$: if we
remove from $H_{1}(\phi ,\varepsilon )$ the vertices of $B_{j}\cup
A_{j}\setminus \{v_{j,1}^{\prime }\}$, for every $j\in \{1,2,\ldots ,8m\}$,
we obtain a graph isomorphic to $H_{\phi }$, where $v_{j,1}^{\prime }$ of $%
H_{1}(\phi ,\varepsilon )$ corresponds to vertex $v_{j}$ of $H_{\phi }$, for
every $j\in \{1,2,\ldots ,8m\}$.

Note that, similarly to $H_{\phi }$, the graph $H_{1}(\phi ,\varepsilon )$
has radius~$2$, since $d(v_{0},u)\leq 2$ in $H_{1}(\phi ,\varepsilon )$ for
every vertex $u$ of $H_{1}(\phi ,\varepsilon )-\{v_{0}\}$. We now prove that 
$H_{1}(\phi ,\varepsilon )$ has also diameter~$3$. First note that the
distance between any two vertices of $\cup _{j=1}^{8m}A_{j}\cup
_{j=1}^{8m}B_{j}$ is at most $2$, since they all have $v_{0}$ as common
neighbor. Consider now an arbitrary vertex $z\in \cup _{j=1}^{8m}A_{j}\cup
_{j=1}^{8m}B_{j}$ and an arbitrary vertex $u_{i,j}\in U$ (resp.~$w_{i,j}\in
W $). If $z$ and $u_{i,j}$ are not adjacent in $H_{1}(\phi ,\varepsilon )$,
there exists the path $(z,v_{0},v_{j,1}^{\prime },u_{i,j})$ (resp.~the path $%
(z,v_{0},v_{j,1}^{\prime },w_{i,j})$) of length $3$ between $z$ and $u_{i,j}$
(resp.~$w_{i,j}$). Therefore $d(z,u)\leq 3$ in $H_{1}(\phi ,\varepsilon )$,
for every vertex $z\in \cup _{j=1}^{8m}A_{j}\cup _{j=1}^{8m}B_{j}$ and every
vertex $u\neq z$ of $H_{1}(\phi ,\varepsilon )$. Moreover, since $H_{\phi }$
is an induced subgraph of $H_{1}(\phi ,\varepsilon )$, it follows by Lemma~%
\ref{G-phi-diameter-3-lem} that also $d(u,u^{\prime })\leq 3$ in $H_{1}(\phi
,\varepsilon )$, for every pair of vertices $u,u^{\prime }\notin \cup
_{j=1}^{8m}A_{j}\cup _{j=1}^{8m}B_{j}$ of $H_{1}(\phi ,\varepsilon )$. Thus $%
H_{1}(\phi ,\varepsilon )$ has diameter~$3$.

Recall now that for every $j\in \{1,2,\ldots ,8m\}$, vertex $v_{j}$ of $%
H_{\phi }$ has been replaced by the vertices of $A_{j}\cup B_{j}$ in $%
H_{1}(\phi ,\varepsilon )$. Furthermore, recall that the vertices of $%
A_{j}\cup B_{j}$ induce in $H_{1}(\phi ,\varepsilon )$ a complete bipartite
graph without a perfect matching between $A_{j}$ and $B_{j}$. Therefore
there exists no pair of sibling vertices in $A_{j}\cup B_{j}$, for every $%
j\in \{1,2,\ldots ,8m\}$. Thus, since $H_{\phi }$ is irreducible by Lemma~%
\ref{H-phi-irreducible-lem}, it follows that $H_{1}(\phi ,\varepsilon )$ is
irreducible as well.

We now prove that $H_{1}(\phi ,\varepsilon )$ is $3$-colorable if and only
if $H_{\phi }$ is $3$-colorable. Suppose first that $H_{1}(\phi ,\varepsilon
)$ is $3$-colorable. Then, since $H_{\phi }$ is an induced subgraph of $%
H_{1}(\phi ,\varepsilon )$, it follows immediately that $H_{\phi }$ is also $%
3$-colorable. Suppose first that $H_{\phi }$ is $3$-colorable, and let $\chi 
$ be a proper $3$-coloring of $H_{\phi }$. Assume without loss of generality
that $v_{0}$ is colored red in $\chi $. We will extend $\chi $ into a proper 
$3$-coloring of $H_{1}(\phi ,\varepsilon )$ as follows. Consider the vertex $%
v_{j}$ of $H_{\phi }$, where $1\leq j\leq 8m$. Since $v_{0}$ is colored red
in $\chi $, it follows that $v_{j}$ is colored either blue or green in $\chi 
$. If $v_{j}$ is colored green in $\chi $, then we color in $H_{1}(\phi
,\varepsilon )$ all vertices of $A_{j}$ green and all vertices of $B_{j}$
blue. Otherwise, if $v_{j}$ is colored blue in $\chi $, then we color in $%
H_{1}(\phi ,\varepsilon )$ all vertices of $A_{j}$ blue and all vertices of $%
B_{j}$ green. It is now straightforward to check that the resulting $3$%
-coloring of $H_{1}(\phi ,\varepsilon )$ is proper, i.e.~that $H_{1}(\phi
,\varepsilon )$ is $3$-colorable. That is, $H_{1}(\phi ,\varepsilon )$ is $3$%
-colorable if and only if $H_{\phi }$ is $3$-colorable. Therefore Theorem~%
\ref{3-col-diam-3-SAT-thm} implies that the formula $\phi $ is satisfiable
if and only if $H_{1}(\phi ,\varepsilon )$ is $3$-colorable.

By construction, the graph $H_{1}(\phi ,\varepsilon )$ has $N={2({n+5m}%
)8m+8m\cdot 2k_{0}+1}$ vertices, where $k_{0}=m^{\varepsilon _{0}}$. Thus,
since $m=\Omega (n)$ and $\varepsilon _{0}\geq 1$, it follows that $N=\Theta
(m^{1+\varepsilon _{0}})$. Therefore $m=\Theta (N^{\frac{1}{1+\varepsilon
_{0}}})$, where $N$ is the number of vertices in $H_{1}(\phi ,\varepsilon )$%
. Furthermore, the degree of $v_{0}$ in $H_{1}(\phi ,\varepsilon )$ is~${%
\Theta (m\cdot k_{0})=\Theta (m^{1+\varepsilon _{0}})}$, the degree of every
vertex $v_{j,p}^{\prime }$ in $H_{1}(\phi ,\varepsilon )$ is~${\Theta
(n+m+k_{0})=\Theta (m^{\varepsilon _{0}})}$, the degree of every vertex~$%
v_{j,p}^{\prime \prime }$ in $H_{1}(\phi ,\varepsilon )$ is $\Theta
(k_{0})=\Theta (m^{\varepsilon _{0}})$, and the degree of every vertex~$%
u_{i,j}$ (resp.~$w_{i,j}$) in $H_{1}(\phi ,\varepsilon )$ is~$\Theta
(m+k_{0})=\Theta (m^{\varepsilon _{0}})$. Therefore the minimum degree of $%
H_{1}(\phi ,\varepsilon )$ is~$\delta =\Theta (m^{\varepsilon _{0}})$. Thus,
since~${m=\Theta (N^{\frac{1}{1+\varepsilon _{0}}})}$, it follows that~$%
\delta =\Theta (N^{\frac{\varepsilon _{0}}{1+\varepsilon _{0}}})$, i.e.~$%
\delta =\Theta (N^{\varepsilon })$.

Summarizing, for every ${\varepsilon \in \lbrack \frac{1}{2},1)}$ and for
every formula $\phi $ with $n$ variables and $m$ clauses, we can construct
in polynomial time a graph $H_{1}(\phi ,\varepsilon )$ with ${N=\Theta
(m^{1+\varepsilon _{0}})}$ (i.e.~$N={\Theta (m^{\frac{1}{1-\varepsilon }})}$%
) vertices and minimum degree~$\delta =\Theta (m^{\varepsilon _{0}})$ (i.e.~$%
{\delta ={\Theta (N^{\varepsilon })}}$), such that $H_{1}(\phi ,\varepsilon
) $ is $3$-colorable if and only if $\phi $ is satisfiable. Moreover, the
constructed graph $H_{1}(\phi ,\varepsilon )$ is irreducible and
triangle-free, and it has diameter~$3$ and radius~$2$. This completes the
proof of the~theorem.\qed
\end{proof}

\medskip

In the next theorem we provide an amplification of the reduction of Theorem~%
\ref{3-col-diam-3-SAT-thm} to the case of graphs with minimum degree ${%
\delta =\Theta (|V|^{\varepsilon })}$, where ${\varepsilon \in \lbrack 0,%
\frac{1}{2})}$.

\begin{theorem}
\label{NP-complete-3-col-diam-3-small-min-degree-thm}Let ${G=(V,E)}$ be an
irreducible and triangle-free graph with diameter~$3$ and radius~$2$. If the
minimum degree of $G$ is ${\delta (G)=\Theta (|V|^{\varepsilon })}$, where ${%
\varepsilon \in \lbrack 0,\frac{1}{2})}$, then it is NP-complete to decide
whether~$G$ is~$3$-colorable.
\end{theorem}

\begin{proof}
Let $G_{1}=(V_{1},E_{1})$ be an arbitrary irreducible and triangle-free
graph with diameter~$3$ and radius~$2$, such that $G_{1}$ has $n$ vertices
and minimum degree $\delta (G_{1})=\Theta (\sqrt{n})$. Note that such a
graph $G_{1}$ exists by the construction of $H_{\phi }$ in Section \ref%
{diam-3-NP-c-subsec} (see also Theorems \ref{3-col-diam-3-SAT-thm} and~\ref%
{NP-complete-3-col-diam-3-large-min-degree-thm}). For simplicity, we
arbitrarily enumerate the vertices of $G_{1}$ as $v_{1},v_{2},\ldots ,v_{n}$%
. Let ${\varepsilon \in \lbrack 0,\frac{1}{2})}$. Define $\varepsilon _{0}=%
\frac{\varepsilon }{1-\varepsilon }$ and $k_{0}=n^{\varepsilon _{0}}$. Since 
${\varepsilon \in \lbrack 0,\frac{1}{2})}$ by assumption, it follows that $%
\varepsilon _{0}\in \lbrack 0,1)$. We now construct from the graph $G_{1}$
the irreducible graph $G_{2}(\varepsilon )$ with diameter~$3$ and radius~$2$
as follows. For every $i\in \{1,2,\ldots ,n\}$, we add the $2k_{0}$ new
vertices $A_{i}=\{v_{i,1}^{\prime },v_{i,2}^{\prime },\ldots
,v_{i,k_{0}}^{\prime }\}$ and $B_{i}=\{v_{i,1}^{\prime \prime
},v_{i,2}^{\prime \prime },\ldots ,v_{i,k_{0}}^{\prime \prime }\}$. The
vertices of $A_{i}$ and of $B_{i}$ induce two independent sets in $%
G_{2}(\varepsilon )$. We then add between the vertices of $A_{i}$ and of $%
B_{i}$ all possible edges, except those of $\{v_{i,\ell }^{\prime }v_{i,\ell
}^{\prime \prime }:1\leq \ell \leq k_{0}\}$. That is,we add $k_{0}^{2}-k_{0}$
edges between the vertices of $A_{i}$ and $B_{i}$, such that they induce a
complete bipartite graph without a perfect matching between $A_{i}$ and $%
B_{i}$. Furthermore, for every $i\in \{1,2,\ldots ,n\}$, we add all possible 
$k_{0}$ edges between $v_{i}$ and the vertices of $A_{i}$. Finally, we
introduce a new vertex $v_{0}$ and we add all $n\cdot 2k_{0}$ possible edges
between $v_{0}$ and the vertices of $\cup _{i=1}^{n}A_{i}\cup
_{i=1}^{n}B_{i} $. Denote the resulting graph by $G_{2}(\varepsilon )$.

Note that, by construction, $G_{1}$ is an induced subgraph of $%
G_{2}(\varepsilon )$. Furthermore, $G_{2}(\varepsilon )$ has $N=n+n\cdot
2k_{0}+1$ vertices, and thus $N=\Theta (n^{1+\varepsilon _{0}})$. Therefore $%
n=\Theta (N^{\frac{1}{1+\varepsilon _{0}}})$, where $N$ is the number of
vertices in $G_{2}(\varepsilon )$. Furthermore, the degree of $v_{0}$ in $%
G_{2}(\varepsilon )$ is~${\Theta (n\cdot k_{0})=\Theta (n^{1+\varepsilon
_{0}})}$, the degree of every vertex $v_{i,\ell }^{\prime }$ (resp.~$%
v_{i,\ell }^{\prime \prime }$) in $G_{2}(\varepsilon )$ is~${\Theta
(k_{0})=\Theta (n^{\varepsilon _{0}})}$, the degree of every vertex~$%
v_{i}\in V_{1}$ in $G_{2}(\varepsilon )$ is at least $\delta
(G_{1})+k_{0}=\Theta (\sqrt{n}+n{^{\varepsilon _{0}}})$. Therefore, for
every $\varepsilon _{0}\in \lbrack 0,1)$, the minimum degree of $%
G_{2}(\varepsilon )$ is $\delta =\Theta ({n^{\varepsilon _{0}}})$. Thus,
since~$n=\Theta (N^{\frac{1}{1+\varepsilon _{0}}})$, it follows that~$\delta
=\Theta (N^{\frac{\varepsilon _{0}}{1+\varepsilon _{0}}})$, i.e.~$\delta
=\Theta (N^{\varepsilon })$.

Note that the graph $G_{2}(\varepsilon )$ has radius~$2$, since $%
d(v_{0},u)\leq 2$ in $G_{2}(\varepsilon )$ for every vertex $u$ of $%
G_{2}(\varepsilon )-\{v_{0}\}$. We now prove that $G_{2}(\varepsilon )$ has
also diameter~$3$. First note that the distance between any two vertices of $%
\cup _{i=1}^{n}A_{i}\cup _{i=1}^{n}B_{i}$ is at most $2$, since they all
have $v_{0}$ as common neighbor. Consider now an arbitrary vertex $z\in \cup
_{i=1}^{n}A_{i}\cup _{i=1}^{n}B_{i}$ and an arbitrary vertex $v_{i}\in V_{1}$%
, where $1\leq i\leq n$. If $z$ and $v_{i}$ are not adjacent in $%
G_{2}(\varepsilon )$, there exists the path $(z,v_{0},v_{j,1}^{\prime
},v_{i})$ of length $3$ between $z$ and $v_{i}$. Therefore $d(z,u)\leq 3$ in 
$G_{2}(\varepsilon )$, for every vertex $z\in \cup _{i=1}^{n}A_{i}\cup
_{i=1}^{n}B_{i}$ and every vertex $u$ of $G_{2}(\varepsilon )-\{z\}$.
Moreover, since $diam(G_{1})=3$ by assumption and $G_{1}$ is an induced
subgraph of $G_{2}(\varepsilon )$, it follows that also $d(v_{i},v_{j})\leq
3 $ in $G_{2}(\varepsilon )$ for every pair of vertices $v_{i},v_{j}\in
V_{1} $. Therefore $G_{2}(\varepsilon )$ has diameter~$3$.

Recall that for every $i\in \{1,2,\ldots ,n\}$ the vertices of $A_{i}\cup
B_{i}$ induce in $G_{2}(\varepsilon )$ a complete bipartite graph without a
perfect matching between $A_{i}$ and $B_{i}$. Therefore there exists no pair
of sibling vertices in $A_{i}\cup B_{i}$, for every $i\in \{1,2,\ldots ,n\}$%
. Thus, since $G_{1}$ is irreducible by assumption, it follows that $%
G_{2}(\varepsilon )$ is irreducible as well.

We now prove that $G_{2}(\varepsilon )$ is $3$-colorable if and only if $%
G_{1}$ is $3$-colorable. If $G_{2}(\varepsilon )$ is $3$-colorable, then
clearly $G_{1}$ is also $3$-colorable, since $G_{1}$ is an induced subgraph
of $G_{2}(\varepsilon )$. Suppose first that $G_{1}$ is $3$-colorable, and
let $\chi $ be a proper $3$-coloring of $G_{1}$ that uses the colors red,
blue, and green. We will extend $\chi $ into a proper $3$-coloring $\chi
^{\prime }$ of $G_{2}(\varepsilon )$ as follows. Consider the vertex $%
v_{i}\in V_{1}$, where $1\leq i\leq n$. If $v_{i}$ is colored red or blue in 
$\chi $, then we color all vertices of $A_{i}$ green and all vertices of $%
B_{i}$ blue in $\chi ^{\prime }$. Otherwise, if $v_{i}$ is colored green in $%
\chi $, then we color all vertices of $A_{i}$ blue and all vertices of $%
B_{i} $ green in $\chi ^{\prime }$. Finally, we color $v_{0}$ red in $\chi
^{\prime }$. It is now straightforward to check that the resulting $3$%
-coloring $\chi ^{\prime }$ of $G_{2}(\varepsilon )$ is proper, i.e.~that $%
G_{2}(\varepsilon )$ is $3$-colorable. That is, $G_{2}(\varepsilon )$ is $3$%
-colorable if and only if $G_{1}$ is $3$-colorable.

Summarizing, for every irreducible graph $G_{1}$ with diameter~$3$ and
radius~$2$, such that $G_{1}$ has $n$ vertices and minimum degree $\delta
(G_{1})=\Theta (\sqrt{n})$, we can construct in polynomial time a graph $%
G_{2}(\varepsilon )$ with ${N=\Theta (n^{1+\varepsilon _{0}})}$ (i.e.~$N={%
\Theta (n^{\frac{1}{1-\varepsilon }})}$) vertices and minimum degree ${%
\delta ={\Theta (N^{\varepsilon })}}$, such that $G_{2}(\varepsilon )$ is $3$%
-colorable if and only if $G_{1}$ is $3$-colorable. Moreover, the
constructed graph $G_{2}(\varepsilon )$ is irreducible and has diameter~$3$
and radius~$2$. This completes the proof of the~theorem, since it is
NP-complete to decide whether $G_{1}$ is $3$-colorable by Theorem~\ref%
{NP-complete-3-col-diam-3-large-min-degree-thm}.\qed
\end{proof}

\medskip

Therefore, Theorems~\ref{NP-complete-3-col-diam-3-large-min-degree-thm} and~%
\ref{NP-complete-3-col-diam-3-small-min-degree-thm} imply that, for every $%
\varepsilon \in \lbrack 0,1)$, the $3$-coloring problem remains NP-complete
for irreducible and triangle-free graphs $G=(V,E)$ with diameter~$3$ and
radius~$2$, where the minimum degree ${\delta (G)}$ is ${\Theta
(|V|^{\varepsilon })}$. However, Theorems~\ref%
{NP-complete-3-col-diam-3-large-min-degree-thm} and~\ref%
{NP-complete-3-col-diam-3-small-min-degree-thm} do not provide any
information about how efficiently (although not polynomially, assuming P$%
\neq $NP) we can decide $3$-coloring on such graphs. In the context of
providing lower bounds for the time complexity of solving NP-complete
problems, Impagliazzo, Paturi, and Zane formulated the \emph{Exponential
Time Hypothesis (ETH)}~\cite{Impagliazzo-Strongly-ETH-01}, which seems to be
very reasonable with the current state of art:

\vspace{0.2cm} \noindent \textbf{Exponential Time Hypothesis (ETH)}~\cite%
{Impagliazzo-Strongly-ETH-01}: There exists no algorithm solving $3$SAT in
time $2^{o(n)}$, where $n$ is the number of variables in the input CNF
formula. \vspace{0.2cm}

Furthermore, they proved the celebrated \emph{Sparsification Lemma}~\cite%
{Impagliazzo-Strongly-ETH-01}, which has the following theorem as a direct
consequence. This result is quite useful for providing lower bounds assuming
ETH, as it parameterizes the running time with the size of the input CNF
formula, rather than only the number of its variables.

\begin{theorem}[\hspace{-0.0005cm}\protect\cite{Impagliazzo-ETH-01}]
\label{ETH-parameter-m-thm}$3$SAT can be solved in time $2^{o(n)}$ if and
only if it can be solved in time $2^{o(m)}$, where $n$ is the number of
variables and $m$ is the number of clauses in the input CNF~formula.
\end{theorem}

The following very well known theorem about the $3$-coloring problem is
based on the fact that there exists a standard polynomial-time reduction
from Not-All-Equal-$3$-SAT to $3$-coloring\footnote{%
Note that there exists a polynomial-time reduction from $3$SAT to
Not-All-Equal-$3$-SAT, such that the size of the output monotone formula is
linear to the size of the input formula.\vspace{-0.2cm}}~\cite%
{PapadimitriouComplexity94}, in which the constructed graph has diameter~$4$
and radius~$2$, and its number of vertices is linear in the size of the
input formula~(see also~\cite{Survey-ETH-11}).

\begin{theorem}[\hspace{-0.0005cm}\protect\cite%
{PapadimitriouComplexity94,Survey-ETH-11}]
\label{no-subexp-diam-4-thm}Assuming ETH, there exists no $2^{o(n)}$ time
algorithm for $3$-coloring on graphs~$G$ with diameter~$4$, radius~$2$, and $%
n$ vertices.
\end{theorem}

We provide in the next three theorems subexponential lower bounds for the
time complexity of $3$-coloring on irreducible and triangle-free graphs with
diameter~$3$ and radius~$2$. Moreover, the lower bounds provided in Theorem~%
\ref{ETH-large-min-degree-thm} are asymptotically almost tight, due to the
algorithm of Theorem~\ref{min-delta-Delta-n-over-delta-alg-diam-3-thm}.

\begin{theorem}
\label{ETH-large-min-degree-thm}Let $\varepsilon \in \lbrack \frac{1}{2},1)$%
. Assuming ETH, there exists no algorithm with running time $2^{o(\frac{N}{%
\delta })}=2^{o(N^{1-\varepsilon })}$ for $3$-coloring on irreducible and
triangle-free graphs $G$ with diameter~$3$, radius~$2$, and $N$ vertices,
where the minimum degree of $G$ is $\delta (G)=\Theta (N^{\varepsilon })$.
\end{theorem}

\begin{proof}
Let $\varepsilon \in \lbrack \frac{1}{2},1)$ and define $\varepsilon _{0}=%
\frac{\varepsilon }{1-\varepsilon }$. In the reduction of Theorem~\ref%
{NP-complete-3-col-diam-3-large-min-degree-thm}, given the value of $%
\varepsilon $ and a boolean formula $\phi $ with $n$ variables and $m$
clauses, we constructed a graph $H_{1}(\phi ,\varepsilon )$ with ${N=\Theta
(m^{1+\varepsilon _{0}})}$ vertices and minimum degree $\delta =\Theta
(m^{\varepsilon _{0}})$. Therefore $\delta =\Theta (N^{\frac{\varepsilon _{0}%
}{1+\varepsilon _{0}}})$, i.e.~$\delta =\Theta (N^{\varepsilon })$.
Furthermore the graph $H_{1}(\phi ,\varepsilon )$ is by construction
irreducible and triangle-free, and it has diameter~$3$ and radius~$2$.
Moreover $\phi $ is satisfiable if and only if $H_{1}(\phi ,\varepsilon )$
is $3$-colorable by Theorem~\ref%
{NP-complete-3-col-diam-3-large-min-degree-thm}.

Suppose now that there exists an algorithm $\mathcal{A}$ that, given an
irreducible triangle-free graph $G$ with $N$ vertices, diameter~$3$, radius~$%
2$, and minimum degree $\delta =\Theta (N^{\varepsilon })$ for some $%
\varepsilon \in \lbrack \frac{1}{2},1)$, decides $3$-coloring on $G$ in time 
$2^{o(\frac{N}{\delta })}$. Then $\mathcal{A}$ decides $3$-coloring on input 
$G=H_{1}(\phi ,\varepsilon )$ in time $2^{o(\frac{N}{\delta })}=2^{o(m)}$.
However, since the formula $\phi $ is satisfiable if and only if $H_{1}(\phi
,\varepsilon )$ is $3$-colorable, algorithm $\mathcal{A}$ can be used to
decide the $3$SAT problem in time $2^{o(m)}$, where $m$ is the number of
clauses in the given boolean formula $\phi $. Thus $\mathcal{A}$ can decide $%
3$SAT in time $2^{o(m)}$. This is a contradiction by Theorem~\ref%
{ETH-parameter-m-thm}, assuming ETH. This completes the proof of the theorem.%
\qed
\end{proof}

\begin{theorem}
\label{ETH-medium-min-degree-thm}Let ${\varepsilon \in \lbrack \frac{1}{3},}%
\frac{1}{2}{)}$. Assuming ETH, there exists no algorithm with running time ${%
2^{o(\delta )}=2^{o(N^{\varepsilon })}}$ for $3$-coloring on irreducible and
triangle-free graphs $G$ with diameter~$3$, radius~$2$, and $N$ vertices,
where the minimum degree of $G$ is $\delta (G)=\Theta (N^{\varepsilon })$.
\end{theorem}

\begin{proof}
We provide for the purposes of the proof an amplification of the reduction
of Theorem~\ref{3-col-diam-3-SAT-thm}. In Section~\ref{diam-3-NP-c-subsec},
given a boolean formula $\phi $ with $n$ variables and $m$ clauses, we
construct the graph $H_{\phi }$, which is irreducible and triangle-free by
Lemma~\ref{H-phi-irreducible-lem}. By its construction in Section~\ref%
{diam-3-NP-c-subsec}, the graph $H_{\phi }$ has ${\Theta (m^{2})}$ vertices.
Furthermore, the degree of vertex $v_{0}$ in $H_{\phi }$ is~${\Theta (m)}$.
The degree of every vertex $v_{j}$ in $H_{\phi }$ is~${\Theta (n+m)=\Theta
(m)}$, where $j\in \{1,2,\ldots ,8m\}$. Finally, the degree of every vertex~$%
u_{i,j}$ (resp.~$w_{i,j}$) in $H_{\phi }$ is~$\Theta (m)$, where $i\in
\{1,2,\ldots ,n+5m\}$ and $j\in \{1,2,\ldots ,8m\}$. Therefore the minimum
degree of $H_{\phi }$ is~$\delta =\Theta (m)$.

Let ${\varepsilon \in \lbrack \frac{1}{3},}\frac{1}{2}{)}$ and define $%
\varepsilon _{0}=\frac{1}{\varepsilon }-2$. Note that $\varepsilon _{0}\in
(0,1]$, since ${\varepsilon \in \lbrack \frac{1}{3},}\frac{1}{2}{)}$. We
construct now from $H_{\phi }$ the graph $H_{2}(\phi ,\varepsilon )$ as
follows. Consider the rows $\ell _{i}$ and $\ell _{i}^{\prime }$ of the
matrix arrangement of the vertices $U\cup W$ of $H_{\phi }$, where $i\in
\{1,2,\ldots ,n+5m\}$. For every $i\in \{1,2,\ldots ,n+5m\}$, we extend row $%
\ell _{i}$ by the new $(m^{1+\varepsilon _{0}}-8m)$ vertices $%
\{u_{i,8m+1},u_{i,8m+2},\ldots ,u_{i,m^{1+\varepsilon _{0}}}\}$. Denote the
resulting row of the matrix with the $m^{1+\varepsilon _{0}}$ vertices $%
\{u_{i,1},u_{i,2},\ldots ,u_{i,m^{1+\varepsilon _{0}}}\}$ by $\widehat{\ell
_{i}}$. Similarly, we extend row $\ell _{i}^{\prime }$ by the new $%
(m^{1+\varepsilon _{0}}-8m)$ vertices $\{w_{i,8m+1},w_{i,8m+2},\ldots
,w_{i,m^{1+\varepsilon _{0}}}\}$. Denote the resulting row of the matrix
with the $m^{1+\varepsilon _{0}}$ vertices $\{w_{i,1},w_{i,2},\ldots
,w_{i,m^{1+\varepsilon _{0}}}\}$ by $\widehat{\ell _{i}^{\prime }}$.
Furthermore, add all necessary edges between vertices of $\widehat{\ell _{i}}
$ and of $\widehat{\ell _{i}^{\prime }}$, such that they induce a complete
bipartite graph. For simplicity of the notation, denote by $U^{\prime }={%
\{u_{i,j}:1\leq i\leq n+5m,1\leq j\leq m^{1+\varepsilon _{0}}\}}$ and ${W}%
^{\prime }{=\{w_{i,j}:1\leq i\leq n+5m,1\leq j\leq m^{1+\varepsilon _{0}}\}}$
the vertex sets that extend the sets $U$ and $W$, respectively. Moreover,
similarly to the notation of Section~\ref{diam-3-NP-c-subsec}, denote $%
U_{j}^{\prime }=\{u_{1,j},u_{2,j},\ldots ,u_{n+5m,j}\}$ and $W_{j}^{\prime
}=\{w_{1,j},w_{2,j},\ldots ,w_{n+5m,j}\}$, for every $j\in \{1,2,\ldots
,m^{1+\varepsilon _{0}}\}$. Then, the vertices of $U_{j}^{\prime }\cup
W_{j}^{\prime }$ contain the vertices of the $j$th column in the (updated)
matrix arrangement of the vertices of $U^{\prime }\cup W^{\prime }$.
Similarly to the construction of the graph $H_{\phi }$ (cf.~Section~\ref%
{diam-3-NP-c-subsec}), we add for every $j\in \{8m+1,8m+2,\ldots
,m^{1+\varepsilon _{0}}\}$ a vertex $v_{j}$ that is adjacent to $%
U_{j}^{\prime }\cup W_{j}^{\prime }\cup \{v_{0})$. Denote the resulting
graph by~${H_{2}(\phi ,\varepsilon )}$.

Now, following exactly the same argumentation as in the proof of Lemma~\ref%
{H-phi-irreducible-lem} and Theorem~\ref{3-col-diam-3-SAT-thm}, we can prove
that: (a)~$H_{2}(\phi ,\varepsilon )$ has diameter~$3$ and radius~$2$, (b)~$%
H_{2}(\phi ,\varepsilon )$ is irreducible and triangle-free, and (c)~the
formula $\phi $ is satisfiable if and only if $H_{2}(\phi ,\varepsilon )$ is 
$3$-colorable.

By the above construction, the graph $H_{2}(\phi ,\varepsilon )$ has $%
N=2(n+5m)\cdot m^{1+\varepsilon _{0}}+m^{1+\varepsilon _{0}}+1$ vertices.
Thus, since $m=\Omega (n)$, it follows that $N=\Theta (m^{2+\varepsilon
_{0}})$. Furthermore, the degree of vertex $v_{0}$ in $H_{2}(\phi
,\varepsilon )$ is~${\Theta (m^{1+\varepsilon _{0}})}$. The degree of every
vertex $v_{j}$ in $H_{2}(\phi ,\varepsilon )$ is~${\Theta (n+m)=\Theta (m)}$%
, where $j\in \{1,2,\ldots ,{m^{1+\varepsilon _{0}}}\}$. Finally, the degree
of every vertex~$u_{i,j}$ (resp.~$w_{i,j}$) in $H_{2}(\phi ,\varepsilon )$
is~$\Theta ({m^{1+\varepsilon _{0}}})$, where $i\in \{1,2,\ldots ,n+5m\}$
and $j\in \{1,2,\ldots ,{m^{1+\varepsilon _{0}}}\}$. Thus the minimum degree
of $H_{2}(\phi ,\varepsilon )$ is~$\delta =\Theta (m)$. Therefore, since $%
N=\Theta (m^{2+\varepsilon _{0}})$, it follows that $\delta =\Theta (N^{%
\frac{1}{2+\varepsilon _{0}}})=\Theta (N^{\varepsilon })$, where $N$ is the
number of vertices in the graph $H_{2}(\phi ,\varepsilon )$.

Suppose now that there exists an algorithm $\mathcal{A}$ that, given an
irreducible triangle-free graph $G$ with $N$ vertices, diameter~$3$, radius~$%
2$, and minimum degree $\delta =\Theta (N^{\varepsilon })$ for some~$%
\varepsilon \in {[\frac{1}{3},}\frac{1}{2}{)}$, decides $3$-coloring on~$G$
in time $2^{o(\delta )}$. Then $\mathcal{A}$ decides $3$-coloring on input $%
G=H_{2}(\phi ,\varepsilon )$ in time $2^{o(\delta )}=2^{o(m)}$. However,
since the formula $\phi $ is satisfiable if and only if $H_{2}(\phi
,\varepsilon )$ is $3$-colorable, algorithm $\mathcal{A}$ can be used to
decide the $3$SAT problem in time $2^{o(m)}$, where $m$ is the number of
clauses in the given boolean formula $\phi $. Therefore $\mathcal{A}$ can
decide $3$SAT in time $2^{o(m)}$. This is a contradiction by Theorem~\ref%
{ETH-parameter-m-thm}, assuming ETH. This completes the proof of the theorem.%
\qed
\end{proof}

\begin{theorem}
\label{ETH-small-min-degree-thm}Let ${\varepsilon \in \lbrack 0,\frac{1}{3})}
$. Assuming ETH, there exists no algorithm with running time ${2^{o(\sqrt{%
\frac{N}{\delta }})}=2^{o(N^{(\frac{1-\varepsilon }{2})})}}$ for $3$%
-coloring on irreducible and triangle-free graphs $G$ with diameter~$3$,
radius~$2$, and $N$ vertices, where the minimum degree of $G$ is $\delta
(G)=\Theta (N^{\varepsilon })$.
\end{theorem}

\begin{proof}
Let $\varepsilon \in \lbrack 0,\frac{1}{3})$ and define $\varepsilon _{0}=%
\frac{\varepsilon }{1-\varepsilon }$. In the reduction of Theorem~\ref%
{NP-complete-3-col-diam-3-small-min-degree-thm}, given the value of~$%
\varepsilon $ and an irreducible graph $G_{1}$ with diameter~$3$ and radius~$%
2$, such that $G_{1}$ has $n$ vertices and minimum degree~${\delta
(G_{1})=\Theta (\sqrt{n})}$, we constructed an irreducible graph $%
G_{2}(\varepsilon )$ with~${N=\Theta (n^{1+\varepsilon _{0}})}$ vertices and
minimum degree~${\delta =\Theta (n^{\varepsilon _{0}})}$. Therefore $\delta
=\Theta (N^{\frac{\varepsilon _{0}}{1+\varepsilon _{0}}})$, i.e.~$\delta
=\Theta (N^{\varepsilon })$. Furthermore the graph~${G_{2}(\varepsilon )}$
has by construction diameter~$3$ and radius~$2$. Moreover $G_{1}$ is $3$%
-colorable if and only if $G_{2}(\varepsilon )$ is $3$-colorable by Theorem~%
\ref{NP-complete-3-col-diam-3-small-min-degree-thm}.

Suppose now that there exists an algorithm $\mathcal{A}$ that, given an
irreducible triangle-free graph $G$ with $N$ vertices, diameter~$3$, radius~$%
2$, and minimum degree $\delta =\Theta (N^{\varepsilon })$ for some $%
\varepsilon \in \lbrack 0,\frac{1}{3})$, decides $3$-coloring on $G$ in time 
$2^{o(\sqrt{\frac{N}{\delta }})}$. Then $\mathcal{A}$ decides $3$-coloring
on input $G=G_{2}(\varepsilon )$ in time $2^{o(\sqrt{\frac{N}{\delta }}%
)}=2^{o(\sqrt{n})}$. However, since $G_{1}$ is $3$-colorable if and only if $%
G_{2}(\varepsilon )$ is $3$-colorable, algorithm $\mathcal{A}$ can be used
to decide $3$-coloring of $G_{1}$ in time $2^{o(\sqrt{n})}$, where $n$ is
the number of vertices in the given graph $G_{1}$. This is a contradiction
by Theorem~\ref{ETH-large-min-degree-thm}, assuming ETH. This completes the
proof of the theorem.\qed
\end{proof}

\section{Concluding remarks\label{conclusions-sec}}

In this paper we investigated graphs with small diameter, i.e.~with diameter
at most~$2$, and at most~$3$. For graphs with diameter at most~$2$, we
provided the first subexponential algorithm for $3$-coloring, with
complexity $2^{O(\sqrt{n\log n})}$, which is asymptotically the same as the
currently best known time complexity for the graph isomorphism problem.
Furthermore we presented a subclass of graphs with diameter~$2$ that admits
a polynomial algorithm for $3$-coloring. An interesting open problem for
further research is to establish the time complexity of $3$-coloring on
arbitrary graphs with diameter~$2$. Moreover, the complexity of $3$-coloring
remains open also for triangle-free graphs of diameter~$2$, or equivalently,
on maximal triangle-free graphs. As these problems (in the general case)
have been until now neither proved to be polynomially solvable nor to be
NP-complete, it would be worth to investigate whether they are polynomially
reducible to/from the graph isomorphism problem.

For graphs with diameter at most~$3$, we established the complexity of $3$%
-coloring, even for triangle-free graphs, which has been a longstanding open
problem. Namely we proved that for every~${\varepsilon \in \lbrack 0,1)}$, $%
3 $-coloring is NP-complete on triangle-free graphs of diameter~$3$ and
radius $2$ with $n$ vertices and minimum degree $\delta =\Theta
(n^{\varepsilon })$. Moreover, assuming~the Exponential Time Hypothesis
(ETH), we provided three different amplification techniques of our hardness
results, in order to obtain for every~${\varepsilon \in \lbrack 0,1)}$
subexponential asymptotic lower bounds for the complexity of $3$-coloring on
triangle-free graphs with diameter~$3$, radius $2$, and minimum degree~${%
\delta =\Theta (n^{\varepsilon })}$. Finally, we provided a $3$-coloring
algorithm with running time~${2^{O(\min \{\delta \Delta ,\ \frac{n}{\delta }%
\log \delta \})}}$ for arbitrary graphs with diameter~$3$, where $n$ is the
number of vertices and $\delta $ (resp.~$\Delta $) is the minimum
(resp.~maximum) degree of the input graph. Due to our lower bounds, the
running time of this algorithm is asymptotically almost tight, when the
minimum degree if the input graph is $\delta =\Theta (n^{\varepsilon })$,
where $\varepsilon \in \lbrack \frac{1}{2},1)$. An interesting problem for
further research is to find asymptotically matching lower bounds for the
complexity of $3$-coloring on graphs with diameter $3$, for all values of
minimum degree $\delta $.

{\small 
\bibliographystyle{abbrv}
\bibliography{ref-diameter}
}

\end{document}